\documentclass[12pt]{article}
\usepackage{amsmath}
\usepackage{amssymb}
\usepackage{amsthm,xspace}
\usepackage[ruled,linesnumbered]{algorithm2e}
\usepackage{xcolor}

\usepackage{fullpage}

\usepackage{mathrsfs}
\usepackage{bbm}
\usepackage{natbib}
\usepackage{hyperref}
\usepackage{graphicx}

\theoremstyle{plain}
\newtheorem{thm}{Theorem}
\newtheorem{lem}[thm]{Lemma}
\newtheorem{prop}[thm]{Proposition}

\theoremstyle{definition}
\newtheorem*{defi}{Definition}

\newtheorem{rem}[thm]{Remark}

\DeclareMathOperator{\diam}{diam}

\DeclareMathOperator{\med}{median}

\DeclareMathOperator*{\argmin}{arg\hbit min}

\newcommand{\var}{\mathrm{Var}}
\newcommand{\eps}{\varepsilon}

\newcommand{\tst}{\textstyle}

\newcommand{\tsum}{{\tst \sum}}

\newcommand{\NN}{\mathbb{N}}

\newcommand{\RR}{\mathbb{R}}

\newcommand{\RD}{\mathbb{R}^D}
\newcommand{\Rplus}{\mathbb{R}_{+}}

\newcommand{\Zplus}{\mathbb{Z}_{+}}
\newcommand{\mcc}{\mathcal{C}}

\newcommand{\mcx}{\mathcal{X}}
\newcommand{\mcy}{\mathcal{Y}}

\newcommand{\bartau}{\bar{\tau}}
\newcommand{\abs}[1]{\lvert #1 \rvert}

\newcommand{\norm}[1]{\lVert #1 \rVert}

\newcommand{\mcn}{\mathcal{N}}
\newcommand{\mcnfin}{\mcn_{\mathrm{fin}}}

\newcommand{\mfm}{\mathfrak{M}}
\newcommand{\mfn}{\mathfrak{N}}
\newcommand{\mfmfin}{\mfm_{\mathrm{fin}}}
\newcommand{\mfnfin}{\mfn_{\mathrm{fin}}}

\newcommand{\tn}{\tilde{n}}

\newcommand{\tz}{\tilde{z}}

\newcommand{\tpi}{\tilde{\pi}}

\newcommand{\txi}{\tilde{\xi}}

\newcommand{\teta}{\tilde{\eta}}

\newcommand{\one}{\mathbbm{1}}
\newcommand{\hbit}{\hspace*{1.5pt}}

\newcommand{\nin}{\noindent}

\newcommand{\khap}{k_{\mathrm{happy}}}
\newcommand{\kmis}{k_{\mathrm{miser}}}
\newcommand{\kal}{k_{\aleph}}

\newenvironment{tightenumerate}
{\begin{list}{\arabic{enumi}.}{\usecounter{enumi} \setlength{\labelwidth}{10mm}
\setlength{\topsep}{1mm} \setlength{\parskip}{1mm}
\setlength{\parsep}{0mm} \setlength{\itemsep}{1mm} 
\setlength{\partopsep}{0mm}}}
{\end{list}}

\newenvironment{tightitemize}
{\begin{list}{$\bullet$}{\usecounter{enumi} \setlength{\labelwidth}{10mm}
\setlength{\topsep}{1mm} \setlength{\parskip}{1mm}
\setlength{\parsep}{0mm} \setlength{\itemsep}{1mm} 
\setlength{\partopsep}{0mm}}}
{\end{list}}

\numberwithin{equation}{section}
\numberwithin{thm}{section}

\xdefinecolor{dgreen}{rgb}{0,0.5,0}
\xdefinecolor{dred}{rgb}{0.6,0,0}

\begin{document}

\title{Metrics and barycenters for point pattern data}
\author{Raoul M\"uller,\footnote{Work supported by DFG RTG 2088.}\ \footnote{Institute for Mathematical Stochastics, University of G\"ottingen, 37077 G\"ottingen, Germany.} \ Dominic Schuhmacher\footnotemark[2] \ and Jorge Mateu\footnote{Work supported by MTM2016-78917-R from the Spanish Ministry of Economy and Competitiveness.}\ \footnote{Department of Mathematics, University Jaume I, 12071 Castell\'on, Spain.}}

\maketitle

\begin{abstract}
  We introduce the transport-transform (TT) and the relative transport-transform (RTT) metrics between finite point patterns on a general space, which provide a unified framework for earlier point pattern metrics, in particular the generalized spike time and the normalized and unnormalized OSPA metrics. Our main focus is on barycenters, i.e.\ minimizers of a $q$-th order Fr\'echet functional with respect to these metrics. 

We present a heuristic algorithm that terminates in a local minimum and is shown to be fast and reliable in a simulation study. The algorithm serves as an umbrella method that can be applied on any state space where an appropriate algorithm for solving the location problem for individual points is available. We present applications to geocoded data of crimes in Euclidean space and on a street network, illustrating that barycenters serve as informative summary statistics. Our work is a first step towards statistical inference in covariate-based models of repeated point pattern observations. 
\medskip

\nin
\textbf{MSC2010 Subject Classification:} Primary 65C60; Secondary 62-07, 90B80.
\smallskip
  
\nin
  \textbf{Key words and phrases:} Fr\'echet mean, Fr\'echet median, network, optimal transport, point process, unbalanced, Wasserstein.
\end{abstract}

\section{Introduction}
\label{sec:intro}

Point pattern data is abundant in modern scientific studies. From biomedical imagery over geo-referrenced disease cases and positions of mobile phone users to climate change related space-time events, such as landslides, we have more and more complicated data available. See \cite{Samartsidis2019}, \cite{KonstantinoudisEtAl2019}, \cite{ChiaraviglioEtAL2016}, \cite{Lombardo2018} for individual examples and
the textbooks \cite{Diggle2013}, \cite{BaddeleyEtAl2015}, \cite{BartekEtAl2018} for a broad overview of further applications. While a few decades ago data consisted typically of a single point pattern in a low dimensional Euclidean space, maybe with some low-dimensional mark information, we have nowadays often multiple observations of point patterns available that may live on more complicated spaces, e.g.\ manifolds (including shape spaces), spaces of convex sets or function spaces. A setting that has received a particularly large amount of attention recently is point patterns on graphs, such as street networks, see \cite{RakshitEtAl2019}, \cite{MoradiEtAl2018} and \cite{MoradiMateu2019} among others.

Multiple point pattern observations may occur by i.i.d.~replication (e.g.\ of a biological experiment), but may also be governed by one or several covariates or form a time series of possibly dependent patterns. Additional mark information can easily be high-dimensional.

Methodology for treating such point pattern data in all these situations is the subject of ongoing statistical research, see e.g.\ \cite{BaddeleyEtAl2015}. From a more abstract point of view, if we think of a point pattern as an element of a metric space $(\mfn, \tau)$, where the metric $\tau$ reflects the concept of distance in an appropriate problem-related way, there is a number of standard methods which can be applied, including multidimensional scaling, disriminant and cluster analysis techniques. This is a stance already taken in \cite{S2014}, Section~1.4, and \cite{MateuEtAl2015}. In the metric space $(\mfn, \tau)$ we can furthermore define a Fr\'echet mean of order $q \geq 1$; that is, for data $\xi_1,\ldots,\xi_k \in \mfn$ any $\zeta \in \mfn$ minimizing
\begin{equation}
  \sum_{j=1}^k \tau(\xi_j,\zeta)^q.
\end{equation}
Such a $q$-th order mean may serve as a ``typical'' element of $\mfn$ to represent the data, and gives rise to more complex statistical analyses, such as Fr\'echet regression; see \cite{PetersenMueller2019} and \cite{LinMueller2019}.

Two metrics on the space of point patterns that have been widely used are the spike time metric, see \cite{VictorPurpura1997} for one dimension and \cite{DiezEtAl2012} for higher dimension, and the Optimal Subpattern Assignment (OSPA) metric, see ~\cite{SXia2008} and~\cite{SVoVo2008}. In the present paper we introduce the \emph{transport--transform (TT) metric} and its normalized version, the \emph{relative transport--transform (RTT) metric}, which provide a unified framework for the earlier metrics. Both the TT- and the RTT-metrics are based on matching the points between two point patterns on a metric space $(\mcx,d)$ optimally in terms of some power $p$ of $d$ and penalizing points that cannot be reasonably matched. We may interpret these metrics as unbalanced $p$-th order Wasserstein metrics, see Remark~\ref{rem:wasserstein} below. In the present paper we always set $p=q$.

Among others \cite{SchoenbergTranbarger2008}, \cite{DiezEtAl2012} and \cite{MateuEtAl2015} have treated Fr\'echet means of order 1 (medians) for the spike time metric under the name of prototypes. However, computations in 2-d and higher were only possible for very small data sets due to a prohibitive computational cost of $O(n^6)$ for the distance between two point patterns with $n$ points each. In the present work we use an adapted auction algorithm that is able to compute TT- and RTT-distances between point patterns in $O(n^3)$. We further provide a heuristic algorithm that bears some resemblance to a $k$-means cluster algorithm and is able to compute local minima of the barycenter problem very efficiently. This makes it possible to compute ``quasi-barycenters'' for 100 patterns of 100 points in $\RR^2$ in a few seconds when basing the TT-distance on the Euclidean distance between points and choosing $p=q=2$.

In Figure~\ref{fig:showcase} we show some typical barycenters obtained by our algorithm in this setting. We use smaller datasets for better visibility. In each scenario there are three different point patterns distinguished by the different symbols in black. The (pseudo-)barycenter represented by the blue circles captures the characteristics of each dataset rather well. Some minor irregularities, especially in the third panel may be due to the fact that only a (good) local optimum is computed. 
\begin{figure}[ht]
  \includegraphics[scale=0.33]{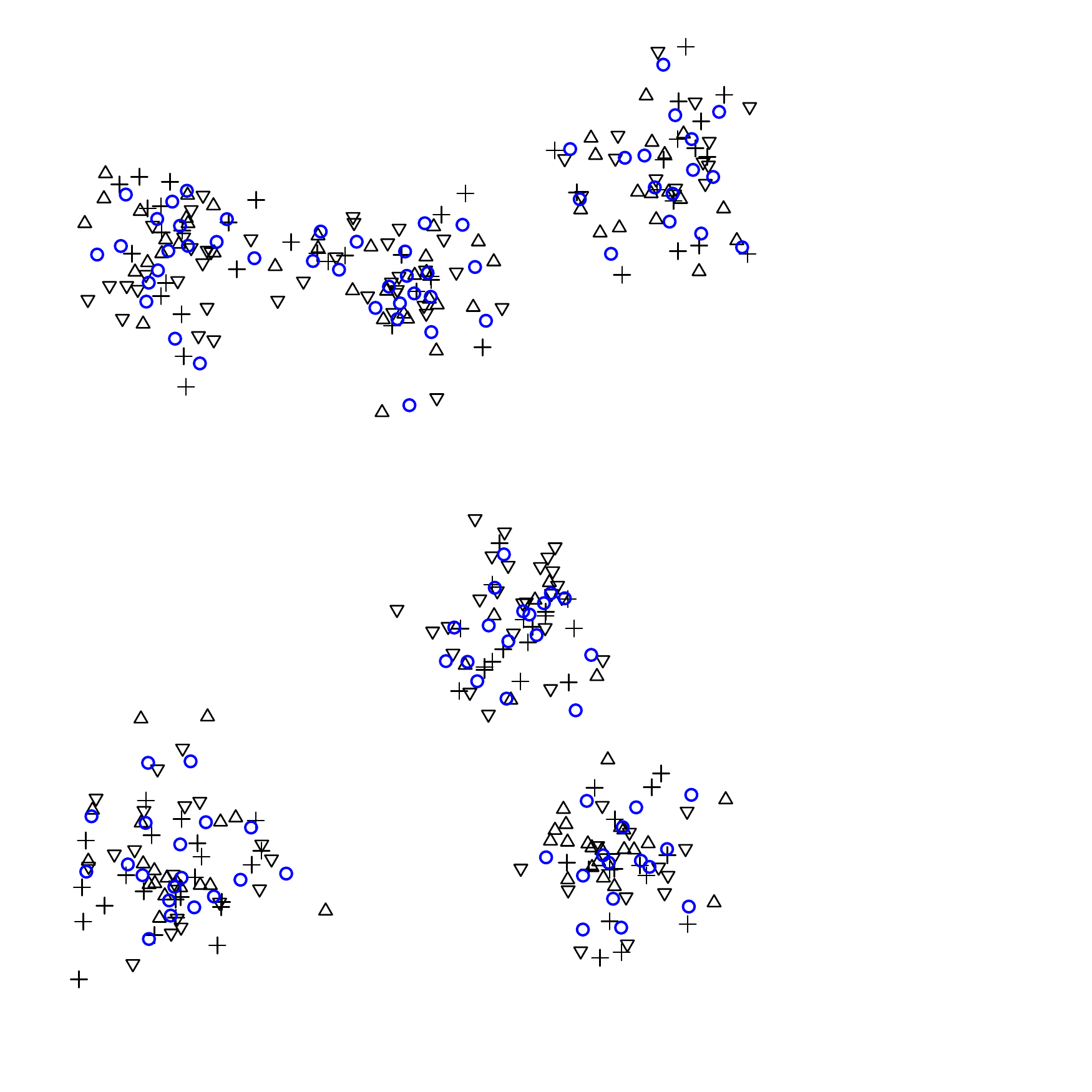}\hspace*{-7mm}\includegraphics[scale=0.33]{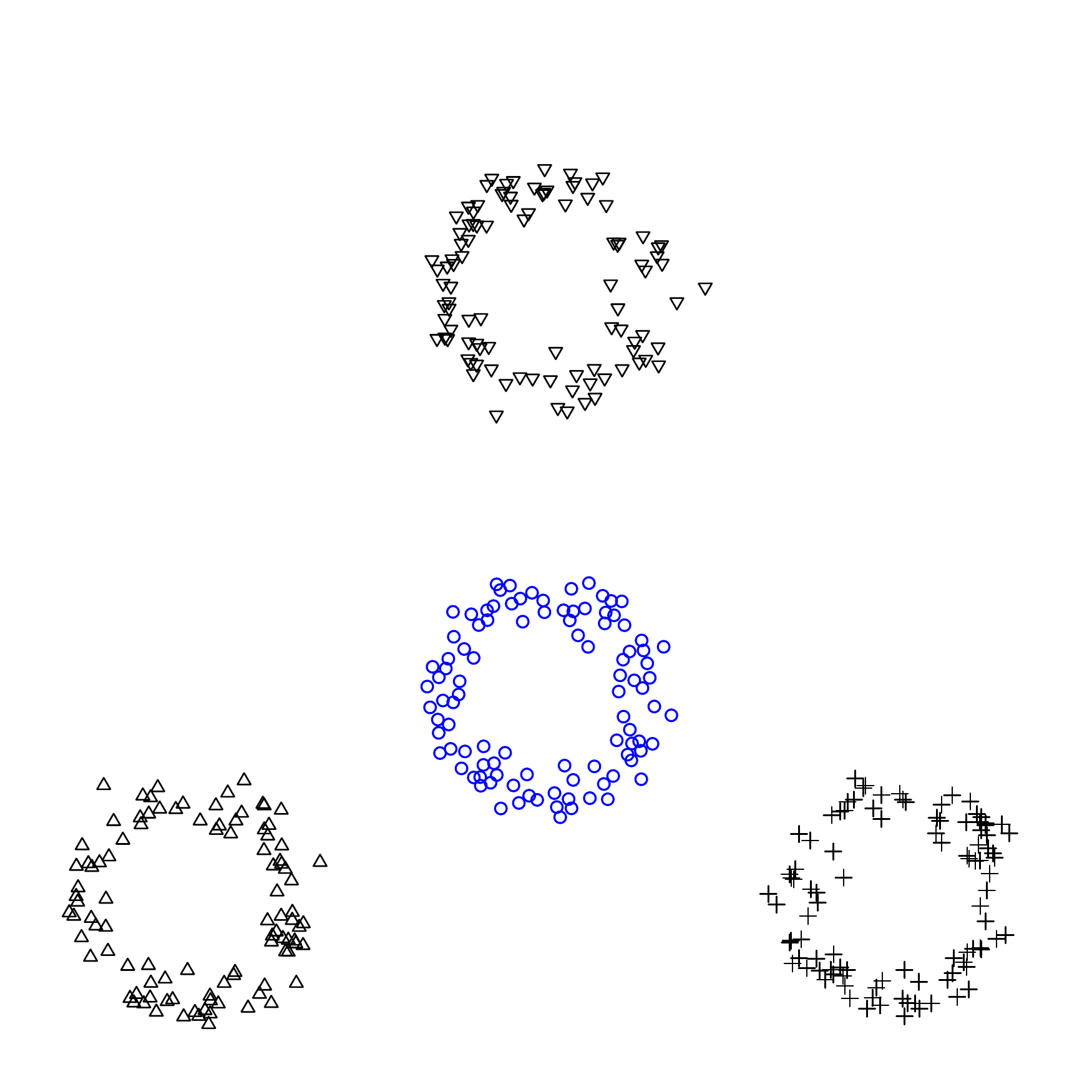}\includegraphics[scale=0.33]{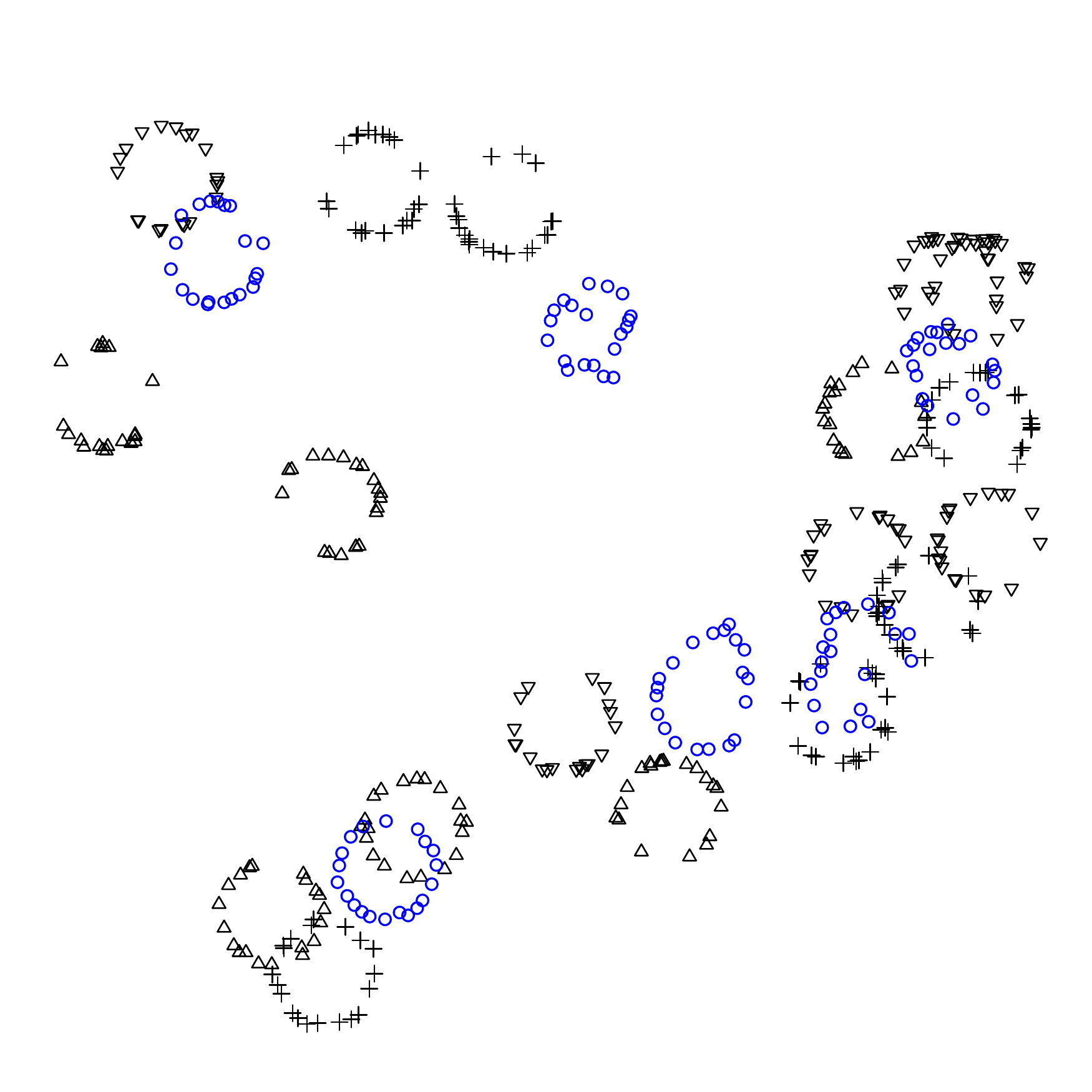}\hspace*{-5mm}
\vspace*{-4mm}
  
  \caption{An example of barycenters computed by our algorithm for three different data sets. In each panel there are three data point patterns indicated by different symbols (black). The resulting (pseudo-)barycenter pattern with respect to Euclidean distance is given by the blue circles ($p=q=2$).}
  \label{fig:showcase}
\end{figure}

More important than being fast for squared Euclidean distance in $\RR^d$ is the fact that our algorithm provides a very general umbrella method, that can in principle be used on \emph{any} underlying space $\mcx$ where an appropriate ``distance function'' between objects is specified as $p$-th power of a metric. All that is required is an algorithm that solves (maybe heuristically) the location problem for individual points in $\mcx$. This allows us e.g.\ to treat the case of point patterns on a network equipped with the shortest-path metric and $p=1$. Figure~\ref{fig:showcase2} gives an example for crime data in Valencia, Spain, which we study in more detail in Section~\ref{sec:applications}.
\begin{figure}[ht]
  \hspace*{2mm}
  \vspace*{-4mm}
  
  \centering
  \includegraphics[width=0.44\linewidth]{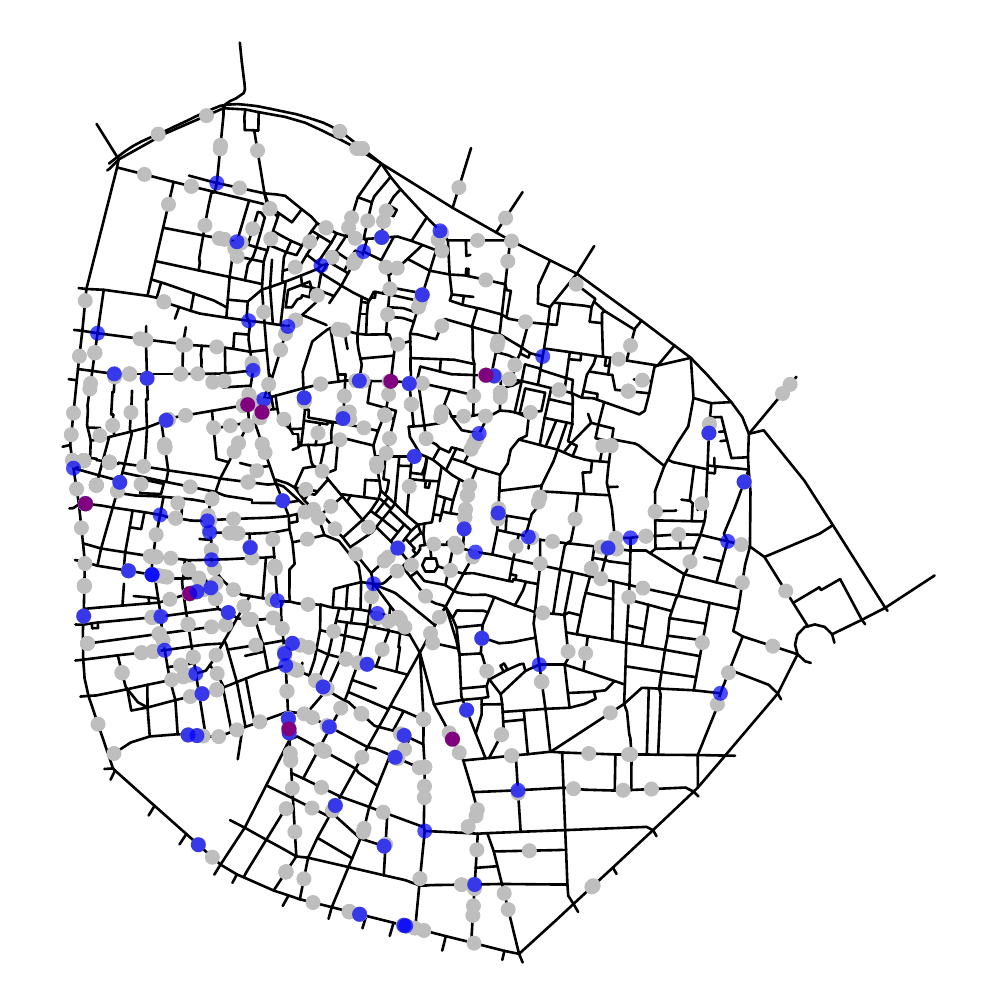}
\vspace*{-5mm}
  
  \caption{An example of a barycenter on a street network. Shown are 8 patterns of assault crimes during the summer months of 2010--2017 in the old town of Valencia (all in grey for better overall visibility). The resulting barycenter with respect to shortest-path distance along the streets is given in blue, with multipoints in purple ($p=q=1$).}
  \label{fig:showcase2}
\end{figure}

The barycenter problem we consider in this paper is closely related to the problem of computing an unbalanced Wasserstein barycenter, see e.g.\ \cite{ChizatEtAl2018}. However, rather than minimizing a Fr\'echet functional on the space of all measures, we minimize on the space of integer valued measures, see Remark~\ref{rem:wasserstein_bary}.

The plan of the paper is as follows. In Section~\ref{sec:ttmetric} we introduce the TT- and RTT-metrics and discuss their relations to spike time, OSPA, and incomplete Wasserstein metrics. Section~\ref{sec:barycenters} specifies what we mean by a barycenter (or Fr\'echet mean) with respect to these metrics and gives an important result that forms the basis for our heuristic algorithm. Two versions of this algorithm, a more direct one and an improved one, which saves computation steps that are unlikely to substantially influence the final result, are discussed in detail in Section~\ref{sec:algorithms}, along with some practical aspects. Section~\ref{sec:simulstudy} contains a larger simulation study, which investigates robustness and runtime performances of the two algorithms for the case of Euclidean distance and $p=2$. Finally, we give two applications to data of crime events on a city map for real data in Section~\ref{sec:applications}. The first one concerns street thefts in Bogot\'a, Colombia. We treat this again as data in Euclidean space, using $p=2$. The second one deals with assault cases in the streets of Valencia, Spain. Here we compute barycenters based on the actual shortest-path distance on the street network and use $p=1$.


  

\section{The transport-transform metric}
\label{sec:ttmetric}

Denote by $\mfnfin$ the space of finite point patterns (counting measures) on a complete separable metric space $(\mcx,d)$, equipped with the usual $\sigma$-algebra $\mcnfin$ generated by the point count maps $\Psi_A\colon \mfnfin \to \RR$, $\xi \mapsto \xi(A)$ for $A \subset \mcx$ Borel measurable.
Elements of $\mfnfin$ are typically denoted by $\xi,\eta,\zeta$ here. As usual we write $\delta_x$ for the Dirac measure with unit mass at $x \in \mcx$. In the present section we mostly use measure notation such as $\xi = \sum_{i=1}^n \delta_{x_i}$, $\xi(\{x\}) \geq 1$ or $\xi + \eta$, but in later sections we also use corresponding (multi)set notation such as $\xi = \{x_1,\ldots,x_n\}$, $x \in \xi$ or $\xi \cup \eta$ where this is unambiguous.

We use $\abs{\xi} = \xi(\mcx)$ to denote the total number of points in the pattern $\xi$.
For $n \in \Zplus = \{0,1,2,\ldots\}$  write $[n] = \{1,2,\ldots,n\}$ (including $[0] = \emptyset$), and denote by $\mfn_n$ the set of point patterns with exactly $n$ points.
We first introduce the metrics we use on $\mfnfin$, which unify and generalize two of the main metrics used previously in the literature.
\begin{defi} ~
  Let $C > 0$ and $p \geq 1$ be two parameters, referred to as \emph{penalty} and \emph{order}, respectively.
  \begin{enumerate}
  \item[(a)] For $\xi = \sum_{i=1}^m \delta_{x_i}, \eta = \sum_{j=1}^n \delta_{y_j} \in \mfnfin$ define the \emph{transport-transform (TT) metric} by  
  \begin{equation}
    \label{eq:ttdef}
      \tau(\xi,\eta) = \tau_{C,p}(\xi,\eta) = \biggl( \min_{(i_1,\ldots,i_l;\hbit j_1,\ldots,j_l) \in S(m,n)} \biggl( (m+n-2l) C^p + \sum_{r=1}^{l} d(x_{i_r},y_{j_r})^p \biggr) \biggr)^{1/p}, 
\end{equation}
where the minimum is taken over equal numbers of pairwise different indices of $[m]$ and $[n]$, respectively, i.e.
\begin{equation*}
\begin{split}
  S(m,n) = \bigl\{ &(i_1,\ldots,i_l;\hbit j_1,\ldots,j_l)\,;\; l \in \{0,1,\ldots,\min\{m,n\}\},\\
    &i_1, \ldots, i_l \in [m] \text{ pairwise different},\, j_1, \ldots, j_l \in [n] \text{ pairwise different} \bigr\}.
\end{split}
\end{equation*}

\item[(b)] For $\xi, \eta \in \mfnfin$ define the \emph{relative transport-transform (RTT) metric} by
  \begin{equation}
    \label{eq:rttdef}
      \bartau(\xi,\eta) = \bartau_{C,p}(\xi,\eta)
    = \frac{1}{\max\{\abs{\xi},\abs{\eta}\}^{1/p}} \hbit \tau_{C,p}(\xi,\eta).
\end{equation}
\end{enumerate}
\end{defi}
We state and prove below that $\tau$ and $\bartau$ are indeed metrics.

The following result simplifies proofs of statements about these metrics and is furthermore invaluable for their computation. The idea is to extend the metric space $(\mcx,d \wedge (2^{1/p} C))$, where $[d \wedge (2^{1/p} C)](x,y) = \min \{ d(x,y), 2^{1/p} C \}$, by setting $\mcx' = \mcx \cup \{\aleph\}$ for an auxiliary element $\aleph \not\in \mcx$ and
  \begin{equation*}
    d'(x,y) = \begin{cases}
       \min \{ d(x,y), 2^{1/p} C \} &\text{if $x,y \in \mcx$;} \\
      C &\text{if $\aleph \in \{x,y\}$ and $x \neq y$;} \\
      0 &\text{if $x=y=\aleph$.}
    \end{cases}
  \end{equation*}
It is shown in Lemma~\ref{lem:metric} that $(\mcx',d')$ is a metric space again. We may then compute distances in the $\tau$ and $\bartau$ metrics by solving an optimal matching problem between point patterns with the same cardinality. For $n \in \NN$ denote by $S_n$ the set of permutations on $[n]$.
\begin{thm} \label{thm:ttassignment}
  Let $\xi = \sum_{i=1}^m \delta_{x_i}, \eta = \sum_{j=1}^n \delta_{y_j} \in \mfnfin$, where w.l.o.g.\ $m \leq n$ (otherwise swap $\xi$ and $\eta$). Set $x_i = \aleph$ for $m+1 \leq i \leq n$ and $\txi = \sum_{i=1}^n \delta_{x_i}$. Then
  \begin{equation*}
    \tau(\xi,\eta) = \biggl( \min_{\pi \in S_n} \sum_{i=1}^n d'(x_i,y_{\pi(i)})^{p} \biggr)^{1/p} \quad \text{and} \quad
    \bartau(\xi,\eta) = \biggl( \frac{1}{n} \min_{\pi \in S_n} \sum_{i=1}^n d'(x_i,y_{\pi(i)})^{p} \biggr)^{1/p}.
  \end{equation*}
\end{thm}
The proof is found in the appendix. 
\begin{rem}[Computation of TT- and RTT-metrics] \label{rem:computation}
  Writing $n$ for the maximum cardinality as in Theorem~\ref{thm:ttassignment}, this result shows that we can compute both $\tau(\xi,\eta)$ and $\bartau(\xi,\eta)$ in worst-time complexity of $O(n^3)$ by using the classic Hungarian method for the assignment problem; see~\cite{Kuhn1955}. In practice we use the auction algorithm proposed in~\cite{Bertsekas1988}, because it has usually much better runtime in our experience, although the default version has a somewhat worse worst-case performance of $O(n^3 \log(n))$.\footnote{There is a modified auction algorithm that can improve the worst-case performance to $O(n^{5/2} \log(n))$; for the performance discussion see~\cite{Bertsekas1988}, page 109. Actually both orders include a factor $c$ in the $\log$ which measures the numerical precision, assumed to be bounded here.}
\end{rem}

\begin{prop} \label{prop:metrics}
  The maps $\tau$ and $\bartau$ are metrics on $\mfnfin$.
\end{prop}
The proof may be found in the appendix.

\begin{prop}[Generalization of previous metrics between point patterns] \label{prop:comparison} ~
  \begin{enumerate}
  \item[(a)] If $p=1$, then for any $\xi,\eta \in \mfnfin$
    \begin{equation}
    \label{eq:spiketime}
      \tau(\xi,\eta) = \min_{(\xi_0, \ldots, \xi_N)} \sum_{i=0}^{N-1} c_{\text{elem}}(\xi_i,\xi_{i+1}),
    \end{equation}
  where the minimum is taken over all $N \in \NN$ and all paths $(\xi_0, \ldots, \xi_N) \in \mfnfin^{N+1}$ such
  that $\xi_0 = \xi$, $\xi_N = \eta$, and from $\xi_i$ to $\xi_{i+1}$ either a single point is added or deleted at cost $c_{\text{elem}}(\xi_i,\xi_{i+1}) = C$ or a single point is moved from $x$ to $y$ at cost $c_{\text{elem}}(\xi_i,\xi_{i+1}) = d(x,y)$.\\
  Thus the TT-metric is the same as the spike time metric (using add and delete penalties $P_a=P_d=C$ and a move penalty $P_m=1$), which was originally introduced on $\Rplus$ by~\cite{VictorPurpura1997} and generalized to metric spaces by~\cite{DiezEtAl2012}.
\item[(b)] If $\diam(\mcx) = \sup_{x,y \in \mcx} d(x,y) \leq 2^{1/p} C$, then for any $\xi = \sum_{i=1}^m \delta_{x_i}, \eta = \sum_{j=1}^n \delta_{y_j} \in \mfnfin$, assuming w.l.o.g.\ $m \leq n$
  \begin{equation}
  \label{eq:ospa}
     \bartau(\xi,\eta)^p = \frac{1}{n} \biggl( (n-m) C^p + \min_{\pi \in S_n} \sum_{i=1}^m d(x_i,y_{\pi(i)})^p \biggr).\\
   \end{equation}
   Thus the RTT-metric is the same as the OSPA metric, introduced in~\cite{SXia2008} and~\cite{SVoVo2008}. Note that in the definition of the OSPA metric $\diam(\mcx) \leq C \leq 2^{1/p} C$ was either required or enforced by taking the minimum of $d$ with $C$. 
   \end{enumerate}
 \end{prop}
 It can be seen from the proof in the appendix that for $p > 1$ the right hand side of~\eqref{eq:spiketime} and for $\diam(\mcx) > 2 C$ 
 the right hand side of~\eqref{eq:ospa} will not be metrics in general.

\begin{rem}[Computation of spike time distances]  \label{rem:spiketime}
  The spike time distances in \cite{VictorPurpura1997} and \cite{DiezEtAl2012} allowed for separate add and delete penalties $P_a$ and $P_d$, as well as a move penalty $P_m$ (factor in front of $d(x,y)$). We set here $P_a=P_d=C$ to obtain a proper metric and divide distances by $P_m$, which is just a scaling. Thus the parameter $C = P_a/P_m = P_d/P_m$ is all that remains.

  \nin
  As noted at the end of Section~4 in \cite{DiezEtAl2012}, having different add and delete penalties may be useful for controlling the total number of points in a barycenter point pattern. Let us point out therefore that Theorem~\ref{thm:ttassignment} is easily adapted to this more general situation by setting $d'(x,y) = \min\{d(x,y), 2^{1/p} (P_a+P_d)\}$, $d'(\aleph,y) = P_a$ and $d'(x,\aleph) = P_d$ for all $x,y \in \mcx$.

  \nin
  In particular this yields a worst-time complexity of $O(n^3)$ for general (maybe asymmetric) spike time distances in general metric spaces, which is a substantial improvement over the $O(n^6)$ complexity of the incremental matching algorithm presented in~\cite{DiezEtAl2012}. 
\end{rem}

\begin{rem}[Unbalanced Wasserstein metrics] \label{rem:wasserstein}
  The TT- and RTT-metrics can be seen as unbalanced Wasserstein metrics, see e.g.~\cite{ChizatEtAl2018}, \cite{LieroEtAl2018} and the references therein. Minimizing over the space $\mfmfin$ of all finite measures on $\mcx \times \mcx$, we obtain the TT-distance as a solution to a particular instance of the unbalanced optimal transport problem in \cite{ChizatEtAl2018}, Definition~2.11, namely
\begin{equation}  \label{eq:chizat}
  \tau(\xi,\eta)^p = \inf_{\gamma \in \mfmfin} \biggl( \int_{\mcx \times \mcx} d(x,y)^p \; \gamma(dx, dy) + C^p \norm{\xi-\gamma_1}_{\mathrm{TV}} + C^p \norm{\eta-\gamma_2}_{\mathrm{TV}} \biggr),
  \end{equation}
where $\gamma_1 = \gamma(\cdot \times \mcx)$ and $\gamma_2 = \gamma(\mcx \times \cdot)$ denote the marginals of $\gamma$, and $\norm{\cdot}_{\mathrm{TV}}$ is the total variation norm of signed measures; specifically $\norm{\mu-\nu}_{\mathrm{TV}} = \sup_{A} (\mu(A)-\nu(A)) + \sup_{A} (\nu(A)-\mu(A))$ for $\mu, \nu \in \mfmfin$, where the suprema are taken over all measurable subsets of $\mcx$.

\nin
Equation~\eqref{eq:chizat} can be shown as follows. It is straightforward to see that we may take the infimum on the right hand side only over $\gamma \in \mfmfin$ with marginals $\gamma_1 \leq \xi$ and $\gamma_2 \leq \eta$, because any additional mass in $\gamma$ may be removed without increasing the total cost of $\gamma$.
Writing $\xi = \sum_{i=1}^n \delta_{x_i}$ and $\eta = \sum_{i=1}^n \delta_{y_i}$ with the help of additional points at $\aleph$ (if necessary), we obtain by similar arguments as in the proof of Theorem~\ref{thm:ttassignment} that the latter problem is equivalent to the discrete transportation problem
  \begin{equation*}
  \begin{split}
    &\min_{(\gamma_{ij})_{1 \leq i,j \leq n}} \sum_{i,j=1}^{n} d'(x_i,y_j)^p \cdot \gamma_{ij} \quad \text{subject to}\\[1mm]
    \tsum_{j = 1}^n \gamma_{ij} = 1 &\ \text{for all $i$}, \quad
    \tsum_{i = 1}^n \gamma_{ij} = 1 \ \text{for all $j$}, \quad
    \gamma_{ij} \geq 0 \ \text{for all $i,j$}.
  \end{split}
\end{equation*}
It is a well-known fact from linear programming that this problem always has a solution $\gamma_{ij} \in \{0,1\}$, $1 \leq i,j \leq n$; see e.g.\ \cite{LuenbergerYe2008}, Section~6.5. We may therefore conclude from Theorem~\ref{thm:ttassignment} that Equation~\eqref{eq:chizat} holds and that the infimum on the right hand side is attained for $\gamma = \sum_{i,j=1}^n \one\{x_i, y_j \neq \aleph\} \gamma_{ij} \delta_{(x_i,y_j)}$.
\end{rem}

In principle, Remark~\ref{rem:wasserstein} allows us to specialize results and algorithms for unbalanced Wasserstein metrics to TT- and RTT-metrics.
However, the discrete setting we consider here is sometimes not included in the general theorems or requires a more specialized treatment. Algorithms for computing unbalanced transport plans are typically derived from balanced optimal transport algorithms; a selection can be found in~\cite{Chizat2017}. The auction algorithm we use in this paper is derived from the auction algorithm used for balanced assignment problems in a similar way.

\section{Barycenters with respect to the TT-metric}
\label{sec:barycenters}

For data on quite general metric spaces, barycenters can formalize the idea of a center element representing the data. In the case of $\mfnfin$ we are thus looking for a center point pattern that gives a good first order representation of a set of data point patterns $\xi_1,\ldots,\xi_k$. More formally we may define a barycenter as the (weighted) $q$-th order Fr\'echet mean with respect to $\tau$; see \cite{Frechet1948}.

\begin{defi}
  For $k \in \NN$ let $\xi_1,\ldots,\xi_k \in \mfnfin$ be data point patterns and $\lambda_1,\ldots,\lambda_k > 0$ with $\sum_{j=1}^k \lambda_j = 1$ be weights. Let furthermore $q \geq 1$. Then we call any
\begin{equation} \label{eq:barydef}
  \zeta_* \in \argmin_{\zeta \in \mfnfin} \sum_{j=1}^k \lambda_j \tau(\xi_j,\zeta)^q
\end{equation}
a \emph{(weighted) barycenter of order $q$}. If no weights are specified we tacitly assume that $\lambda_j = 1/k$ for $1 \leq j \leq k$, leading to an ``unweighted'' barycenter.
\end{defi}

\begin{rem} \label{rem:frechet}
  For $q = 2$ barycenters on general metric spaces are simply known as (empirical) \emph{Fr\'echet means}. For $q=1$ they are sometimes known as \emph{Fr\'echet medians}. This comes from the fact that given $x_1,\ldots,x_k \in \RD$, we have  
\begin{equation}  \label{eq:frechet2}
  \argmin_{z \in \RD} \sum_{j=1}^k \norm{x_j-z}^2 = \frac{1}{k} \sum_{j=1}^k x_j 
\end{equation}
(the $\argmin$ is unique here), and that given $x_1, \ldots,x_k \in \RR$, we have
\begin{equation}  \label{eq:frechet1}
  \argmin_{z \in \RR} \sum_{j=1}^k \norm{x_k-z} = \med \{x_1,\ldots,x_k\},
\end{equation}
where the right hand side denotes the set of medians~$\bigl\{z \in \RR;\, \#\{j;\, x_j \leq z\} = \#\{j;\, x_j \geq z\} \bigr\}$.
\end{rem}

\begin{rem} \label{rem:wasserstein_bary}
  As seen in Remark~\ref{rem:wasserstein} we may interpret $\tau$ as an unbalanced Wasserstein metric. There has been a great deal of research on Wasserstein barycenters (in the Fr\'echet mean sense as above, see e.g.\ \cite{AguehCarlier2011} or \cite{CuturiDoucet2014}), which more recently also extends to unbalanced Wasserstein metrics, see e.g.\ \cite{ChizatEtAl2018} or \cite{SchmitzEtAl2018}. In addition to the fact that much of the corresponding theory is not well adapted to the case of discrete input measures, with the notable exception of \cite{AnderesEtAl2016}, we point out that a fundamental difference of~\eqref{eq:barydef} lies in the fact that we minimize over the space of integer-valued rather than general measures.
\end{rem}
  
In what follows we always set $p = q$ and choose this number mostly $\in \{1,2\}$.
We refer to the resulting barycenters simply as $1$- and \emph{$2$-barycenter} or as \emph{point pattern median} and \emph{point pattern mean}, respectively. Point pattern medians have been introduced under the name of \emph{prototypes} in \cite{SchoenbergTranbarger2008} on $\RR$ and studied in higher dimensions in~\cite{DiezEtAl2012} and~\cite{MateuEtAl2015}. However, in these papers the applicability was limited to rather small datasets due to the large computation cost of $O(n^6)$ mentioned in Remark~\ref{rem:spiketime}.

Using the construction from Theorem~\ref{thm:ttassignment}, we may reformulate the barycenter problem as a multidimensional assignment problem, generalizing Lemma~16 in~\cite{KolianderEtAl2018}. Note that for the TT-metric we can add an arbitrary number of points at $\aleph$ to both point patterns without changing the minimum in Theorem~\ref{thm:ttassignment}.

\begin{prop} \label{prop:mdagengen}
  For point patterns $\xi_j = \sum_{i=1}^{n_j} \delta_{x_{ij}}$, $j \in [k]$, let $\tn := \bigl\lfloor \frac{2}{k+1} \sum_{j=1}^{k} n_j \bigr\rfloor$ and $n \geq \max\{\tn, n_j;\, 1 \leq j \leq k\}$. Set $x_{ij} = \aleph$ for $n_j+1 \leq i \leq n$ and $\txi_j = \sum_{i=1}^{n} \delta_{x_{ij}}$ for any $j \in [k]$.

  Then for any $\pi_{*,1},\ldots,\pi_{*,k} \in S_{n}$ jointly minimizing
  \begin{equation} \label{eq:mdagengen}
    \sum_{i=1}^{n} \min_{z \in \mcx'} \sum_{j=1}^k d'(x_{\pi_j(i),j},z)^p
  \end{equation}
  the point pattern $\zeta_*\vert_{\mcx}$ with $\zeta_* = \sum_{i=1}^{n} \delta_{z_{i}}$, where $z_i \in \argmin_{z \in \mcx'} \sum_{j=1}^k d'(x_{\pi_{*,j}(i),j},z)^p$ is a $p$-th order barycenter with respect to the TT-metric.
\end{prop}
The $\pi_{*,1},\ldots,\pi_{*,k} \in S_{n}$ above define $n$ disjoint ``clusters'' $\mcc_i = \{x_{\pi_{*,j}(i),j};\, 1 \leq j \leq k\}$, where each contains exactly one (maybe virtual) point of each point pattern. The minimization of \eqref{eq:mdagengen} may thus be interpreted as a multidimensional assignment problem with cluster cost
\begin{equation}
  \label{eq:mdaclustcost}
  \mathrm{cost}_*(\mcc) = \min_{z \in \mcx'} \sum_{x \in \mcc} d'(x,z)^p.
\end{equation}
\begin{proof}
Let us first give an upper bound on the cardinality of the barycenter. A single barycenter point can be matched with up to $k$ points (one from each point pattern). If said point is matched with only $\frac{k}{2}$ points or fewer, it cannot be worse to delete it. The contribution for this point in the objective function is at least $\frac{k}{2} C$, while deleting it adds at most $\frac{k}{2} C$ to the objective function. 

So every barycenter point should be matched with at least $\lceil \frac{k+1}{2} \rceil$ points. The total number of points is $\sum_{j=1}^{k} n_j$. Therefore the number of barycenter points is bounded above by $\tilde{n} = \bigl\lfloor \frac{2}{k+1} \sum_{j=1}^{k} n_j \bigr\rfloor$.


It is thus sufficient to fill up all the point patterns $\xi_j$ to $n$ points and work also with an ansatz of $n$ points for $\zeta$. Theorem~\ref{thm:ttassignment} yields
    \begin{equation} \label{eq:mdaexchange}
      \begin{split}
      \min_{\zeta \in \mfnfin} \sum_{j=1}^k \tau(\xi_j,\zeta)^p 
      &= \min_{z_1,\ldots,z_{n} \in \mcx'} \sum_{j=1}^k \min_{\pi \in S_{n}} \sum_{i=1}^{n} d'(x_{\pi(i),j},z_i)^p \\
      &= \min_{z_1,\ldots,z_n \in \mcx'} \min_{\pi_1,\ldots,\pi_k \in S_{n}} \sum_{j=1}^k \sum_{i=1}^{n} d'(x_{\pi_j(i),j},z_i)^p \\
      &= \min_{\pi_1,\ldots,\pi_k \in S_{n}} \sum_{i=1}^{n} \min_{z_i \in \mcx'} \sum_{j=1}^k d'(x_{\pi_j(i),j},z_i)^p
  \end{split}
\end{equation}
and that any minimizer $\zeta_*\vert_{\mcx}=\sum_{i=1}^{n} \delta_{z_i}\vert_{\mcx}$ on the left hand side is obtained from jointly minimizing in $\pi_1,\ldots,\pi_k$ and $z_1,\ldots,z_{n}$ on the right hand side.
\end{proof}

\section{Alternating clustering algorithms}
\label{sec:algorithms}

\SetKwData{it}{it}\SetKwData{N}{N}\SetKwData{pplist}{pplist}\SetKwData{weight}{weight}\SetKwData{cost}{cost}\SetKwData{costold}{costold}\SetKwData{costnew}{costnew}\SetKwData{dist}{dist}\SetKwData{center}{center}\SetKwData{dist2}{dist2}\SetKwData{centertwo}{center2}\SetKwData{perm}{perm}\SetKwData{perminfo}{perminfo}\SetKwData{pena}{penalty}\SetKwData{centerin}{centerIn}\SetKwData{centerout}{centerOut}\SetKwData{centernew}{centernew}\SetKwData{codexi}{xi}\SetKwData{codeeta}{eta}\SetKwData{true}{true}\SetKwData{false}{false}\SetKwData{happycluster}{happycluster}\SetKwData{happypoints}{happypoints}\SetKwData{newhappypoints}{newhappypoints}\SetKwData{cluster}{cluster}\SetKwData{khappy}{khappy}\SetKwData{happycost}{happycost}\SetKwData{supply}{supply}\SetKwData{nexist}{nexist}\SetKwData{kexistnew}{kexistnew}\SetKwData{whichexist}{whichexist}\SetKwData{zproposal}{zproposal}\SetKwData{zproposaltwo}{zproposal2}\SetKwData{newcluster}{newcluster}\SetKwData{newperm}{newperm}\SetKwData{newcost}{newcost}\SetKwData{noncenters}{noncenters}\SetKwData{miserablepoints}{miserablepoints}\SetKwData{alephindex}{alephindex}\SetKwData{nda}{nda}\SetKwData{epsvec}{epsvec}

\SetKwFunction{kassignment}{kassignment}\SetKwFunction{sumppdist}{sumppdist}\SetKwFunction{centerFromPerm}{centerFromPerm}\SetKwFunction{plotAssignment}{plotAssignment}\SetKwFunction{print}{print}
\SetKwFunction{optimPerm}{optimPerm}\SetKwFunction{optimClusterCenter}{optimClusterCenter}\SetKwFunction{optimClusterCenterEuclid2}{optimClusterCenterEuclid2}\SetKwFunction{optimClusterCenterNetwork1}{optimClusterCenterNetwork1}\SetKwFunction{optimBary}{optimBary}\SetKwFunction{optimCenter}{optimCenter}\SetKwFunction{optimDelete}{optimDelete}\SetKwFunction{optimAdd}{optimAdd}\SetKwFunction{warning}{warning}\SetKwFunction{kmeansbary}{kMeansBary}
\SetKwFunction{auctionmatch}{auctionMatch}\SetKwFunction{computecost}{computeCost}\SetKwFunction{matrix}{matrix}\SetKwFunction{getHappyCluster}{getHappyCluster}\SetKwFunction{codexists}{exists}\SetKwFunction{append}{append}\SetKwFunction{mean}{mean}
\SetKwFunction{size}{size}\SetKwFunction{crossdistp}{crossdistp}\SetKwFunction{codesum}{sum}\SetKwFunction{getAllMiserablePoints}{getAllMiserablePoints}\SetKwFunction{getBestCluster}{getBestCluster}\SetKwFunction{sample}{sample}\SetKwFunction{optimizeCluster}{optimizeCluster}

\SetKwInOut{input}{Input}\SetKwInOut{output}{Output}
\SetKw{InSet}{in}

Based on Proposition~\ref{prop:mdagengen} we propose an algorithm that alternates between minimizing
\begin{equation} \label{eq:objective}
  \sum_{j=1}^k \sum_{i=1}^{n} d'(x_{\pi_j(i),j},z_i)^p
\end{equation}
in $\pi_1,\ldots,\pi_k \in S_{n}$ and in $z_1,\ldots,z_{n} \in \mcx'$ until convergence. Such an algorithm terminates in a local minimum of~\eqref{eq:objective}
after a finite number of steps, because \eqref{eq:objective} can never increase and the minimization in the permutations is over a finite space.

Since this underlying idea is similar to the popular $k$-means clustering algorithm, we named the main function in the pseudocode and in the actual implementation \texttt{kMeansBary} (note, however, that $n$ plays the role of $k$ in our notation). A similar algorithm in the quite different setting of trimmed Wasserstein-2 barycenters for finitely supported probability measures is proposed in \cite{delBarrioEtAl2019}.

In what follows we present pseudocode along with the underlying ideas and explanations for two versions of the \texttt{kMeansBary}-algorithm that we dub \emph{original} and \emph{improved}. Here ``improved'' refers to the fact that we cut down on certain computation steps in order to save runtime. We will see in Section~\ref{sec:simulstudy} that this comes essentially without any performance loss.

User-friendly implementations of both algorithms will be made publicly available in an {\sf R}-package.

\subsection{Our original \texttt{kMeansBary} algorithm}
\label{ssec:original}

The pseudocode for the basic alternating strategy described above is given in Algorithm~\ref{algo:kmeansbary}. We have introduced a stopping parameter $\delta$ to allow termination before the local optimum is reached. Since we are not interested in the actual clustering, but only in the position of the centers $z_1,\ldots,z_{n}$, it seems very unlikely (though possible) that the solution changes substantially once the cost decrease has become very small. What is more, such a change might be spurious due to rounding errors in the data or when we use an approximation method for optimizing in the centers. 
Note also that we can always set $\delta$ to the smallest representable positive floating-point number to ensure convergence to the local optimum.

\IncMargin{1em}
\begin{algorithm}[h]
  \caption{\FuncSty{kMeansBary}. Dependence on data \protect\pplist suppressed for simplicity.
}\label{algo:kmeansbary}  
\input{$\center$ an initial pseudo-barycenter;\\
       $\pplist$ the list of data point patterns;\\
       $\delta > 0$ a constant for the termination criterion;\\
       $N$ the maximum number of iterations.}   
\output{Locally optimal pseudo-barycenter $\center$.}
\BlankLine
$\perm, \cost \leftarrow \optimPerm(\center)$\;
\For{$\it \leftarrow 1$ \KwTo $N$}{
  $\costold \leftarrow \cost$\;
  $\center \leftarrow \optimBary(\perm, \center)$\;
  $\center \leftarrow \optimDelete(\perm, \center)$\;
  $\center \leftarrow \optimAdd(\perm, \center)$\;
  $\perm, \cost \leftarrow \optimPerm(\center)$\;
  \lIf(\tcp*[f]{difference always nonnegative}){$\costold-\cost < \delta$}{
    \textbf{break}}
}
\KwRet{$\center$}\tcp*{and warn if the loop has run out}
\end{algorithm}\DecMargin{1em}

The minimization with respect to $\pi_1,\ldots,\pi_k$ is performed by \optimPerm.
This function computes an optimal matching between the current \center and each data point pattern in \pplist, using an alternating version of the auction algorithm with $\eps$-scaling; see Remark~\ref{rem:computation} and \cite{Bertsekas1988} for more details. We output the \cost of the current matching and an $n \times k$ matrix \perm, whose $j$-th column specifies the order in which the points of the $j$-th data pattern are matched to $z_1,\ldots,z_{n}$. For greater efficiency we save auxiliary information (price and profit vectors) and use it for initializing the auction algorithm when calling it again with the same data point pattern. 

For practical purposes we have split up the minimization with respect to $z_1,\ldots,z_{n} \in \mcx'$ into a function \optimBary that optimizes the positions within $\mcx$ and functions \optimDelete and \optimAdd that optimize which of the $z_i$ to move from $\mcx$ to $\aleph$ and from $\aleph$ to $\mcx$, respectively. We discuss details of these functions below.

In addition to the outputs of the various functions shown in Algorithm~\ref{algo:kmeansbary}, we also keep information on the quality of each match of points up to date. We call the match of a $z_i$ with a data point $x_{i'j}$ 
\begin{equation*}
  \begin{split}
    \text{\emph{happy}} &\text{ \,if $z_i, x_{i'j} \in \mcx$ and $d'(z_i,x_{i'j}) < 2^{1/p} C$}\\
    \text{\emph{miserable}} &\text{ \,if $z_i, x_{i'j} \in \mcx$ and $d'(z_i,x_{i'}) = 2^{1/p} C$ or if $z_i = \aleph, x_{i'j} \in \mcx$}\\
    \text{\emph{to $\aleph$}} &\text{ \,if $x_{i'j} = \aleph$}.
   \end{split}
\end{equation*}
Note that a miserable match is worst possible in the sense that $\mathrm{cost}(\mcc_i) = \sum_{x \in \mcc_i} d'(x,z_i)^p$ for center $z_i$ cannot increase if $x_{i'j}$ is replaced by \emph{any} other $x \in \mcx'$.

\subsubsection*{\optimBary}

The purpose of this function is to find for each $z_i \in \mcx$ (i.e.~not currently at $\aleph$) a location in $\mcx$ that minimizes $\mathrm{cost}(\mcc_i)$ for its current cluster $\mcc_i = \{x_{\pi_{j}(i),j};\, 1 \leq j \leq k\}$. This amounts to a more traditional location problem in $\mcx$, except that it is typically made (much) more difficult by the fact that we have to truncate distances at $2^{1/p} C$.

Note that any cluster points at $\aleph$ can be ignored because they always contribute the same amount to the cluster cost, no matter where the center lies. The same is true for individual points that have a much larger distance than $2^{1/p} C$ from the bulk of the points. However, there are countless scenarios with (groups of) points being around distance $2^{1/p} C$ apart from one another for which the optimization of the cluster cost becomes a difficult optimization problem (piecewise smooth on a space that is fragmented in complicated ways).

As a simple heuristic that works well in cases where we do not have to cut too many distances (i.e.\ $C$ is not too small), we suggest to ignore all points that are at the maximal $d'$-distance $2^{1/p} C$ from the current $z_i$ when computing the new $z_i$. Note that in this way the cluster cost can never increase.

Algorithm~\ref{algo:optimbary} gives corresponding pseudocode. The function \optimClusterCenter handles the location problem for the untruncated metric $d$ on $\mcx$. If for example $\mcx = \RD$ equipped with the Euclidean metric and $p=2$, Equation~\eqref{eq:frechet2} implies that \optimClusterCenter simply has to take the (coordinatewise) average of all happy points. The case $p=1$ can be tackled with higher computational effort by approximation via the popular Weiszfeld algorithm; see~\cite{Weiszfeld1937}. 

\IncMargin{1em}
\begin{algorithm}[h]
  \caption{\FuncSty{optimBary}: find optimal center pattern \emph{within $\mcx$} for given clusters $\mcc_i$.\hspace*{-12mm}
}\label{algo:optimbary}  
\For{$i \leftarrow 1$ \KwTo $n$}{
  \If(\tcp*[f]{$z_i$ is the current center of the $i$th cluster $\mcc_i$}){$z_i \in \mcx$}{   
    $\happypoints \leftarrow \{\text{points in $\mcc_i$ that are happily matched to $z_i$}\}$\;
    \If{$\happypoints\ !\!= \emptyset$}{
      $z_i \leftarrow \optimClusterCenter(\happypoints)$\;
    }(\tcp*[f]{otherwise $z_i$ is deleted in next call to \optimDelete})
}
}
\KwRet{$\{ z_1, \ldots, z_{n} \}$}\;
\end{algorithm}\DecMargin{1em}

As a further instance, which we will take up in Section~\ref{sec:applications}, we consider the situation where $\mcx$ is a simple graph $(V,E)$ equipped with the shortest-path distance and $p=1$. It can be shown that in this case the location problem in $\mcx$ is solved by an element $z_i$ of $V \cup \mcc_i$, i.e.\ either a vertex of the graph or any data point, see \cite{Hakimi1964}. We therefore proceed by first computing the distance matrix between all these points, which is then used for the entire algorithm. Such shortest-path distance computations in sparse graphs with thousands of points can be performed in (at most) a few seconds by various algorithms, see Chapter~25 in \cite{CormenetAl2009} and the concrete timing in Subsection~\ref{ssec:valencia}. It is now easy to implement the function \optimClusterCenter. For a given set of happy points of a cluster $\mcc_i$, pick the corresponding columns in the distance matrix, add them up and determine the minimal entry of the resulting vector. If there are several such entries, which due to choosing $p=1$ can happen quite frequently, we pick one among them uniformly at random. The index of the obtained entry identifies the center point $z_i$.

Note that precomputing the distance matrix between all points in $V \cup \mcc_i$ has the additional advantage that no distances have to be computed in the {\optimPerm} step. 

In situations with very large graphs and/or data sets, computing the entire distance matrix may not be feasible. In this case there are various fast heuristics available, such as the single and multi-hub heuristics proposed (in principle) in \cite{BandeltEtAl1994} and \cite{KolianderEtAl2018}.

\subsubsection*{\optimDelete}

This function deletes (i.e.\ moves to $\aleph$) any $z_i \in \mcx$ for which this operation decreases $\mathrm{cost}(\mcc_i)$. 

We denote by $\khap$, $\kmis$ and $\kal$ the numbers of data points in $\mcc_i$ that are happy, miserable and at $\aleph$, respectively. Write furthermore $c_{\mathrm{happy}}$ for the total cost of matching the happy points to $z_i$. If $z_i$ stays in $\mcx$ the cluster incurs an overall total cost of
\begin{equation*}
  c_{\mathrm{happy}} + \kmis \cdot 2C^p + \kal \cdot C^p
 \end{equation*}
 as opposed to
\begin{equation*}
  \khap \cdot C^p + \kmis \cdot C^p
 \end{equation*}
 if we delete $z_i$. Subtracting $\kmis \cdot C^p$ from both expressions this leads to the deletion condition
 \begin{equation*}
  \khap C^p < c_{\mathrm{happy}} + (k-\khap) C^p.
  \end{equation*}
  Since $c_{\mathrm{happy}} \geq 0$, a sufficient condition for deletion is $2 \khap < k$. We use this as a quick pre-test, which allows us to avoid computing $c_{\mathrm{happy}}$ sometimes.
The full deletion procedure is presented in Algorithm~\ref{algo:optimdelete}.

\IncMargin{1em}
\begin{algorithm}[h]
\caption{\FuncSty{optimDelete}: move center points from $\mcx$ to $\aleph$ if it decreases cost.}
\label{algo:optimdelete}  
\For{$i \leftarrow 1$ \KwTo $n$}{
  \If{$z_i \in \mfn$}{  
    $\happypoints \leftarrow \{\text{points in $\mcc_i$ that are happily matched to $z_i$}\}$\;    
    $\khap \leftarrow \# \happypoints$\;
  \eIf{$2 * \khap < k$}{
    $z_i \leftarrow \aleph$\tcp*{shortcut deletion}
  }{
    $c_{\mathrm{happy}} \leftarrow \sum_{x \in \happypoints} d'(x,z_i)^p$\;
    \lIf{$\khap * C^p < c_{\mathrm{happy}} + (k - \khap) * C^p$}{
      $z_i \leftarrow \aleph$}
 }
}
}
\KwRet{$\{ z_1, \ldots, z_n \}$}\;
\end{algorithm}\DecMargin{1em}

\subsubsection*{\optimAdd}

This function adds (i.e.\ moves to $\mcx$) any $z_i \in \aleph$ for which it finds a way to do so that decreases $\mathrm{cost}(\mcc_i)$. Pseudocode is given in Algorithm~\ref{algo:optimadd}.

As a compromise between computational simplicity and finding a good location in $\mcx$, we first sample a proposal location $\tz$ uniformly from all miserable data points (i.e.~points from \emph{any} cluster that are currently in a miserable match with their center). Before we consider moving $z_i$ to $\tz$, we rebuild the cluster $\mcc_i$ in such a way that this move has a better chance of being accepted.

The corresponding procedure is performed by the {\optimizeCluster}-function in the pseudocode:
For each data pattern $\xi_j$ pick the miserable point that is closest to $\tz$ (if it has any) and exchange it with the corresponding point $x_{\pi_j(i),j}$ that is currently in $\mcc_i$.
Since the point coming from the other cluster was miserable before, the cost of that cluster cannot increase by this exchange. The cost of the cluster $\mcc_i$ can increase only if it loses a point located at $\aleph$ in the exchange. In this case the cost increases by $C^p$, which is compensated by the fact that the cost of the other cluster must decrease, either from $2 C^p$ to $C^p$ if its center is in $\mcx$, or from $C^p$ to $0$ if its center is at $\aleph$. Thus the total cost remains the same, but $\mcc_i$ has an additional point in $\mcx$ now, which makes the successful addition of $z_i$ to $\mcx$ more likely. 
%

To further decrease the prospective cluster cost after addition, we update the proposal $\tz$ by recentering it in its new cluster using the appropriate $\optimClusterCenter$-function introduced in $\optimBary$ (applied to the set of points of the new cluster that are in a happy match with $\tz$).

Finally, check whether the cost of the new cluster based on the updated $\tz$ is smaller than the same cost based on $z_i = \aleph$, which is $C^p$ times the number $k_{\mcx}$ of non-$\aleph$ points in the new cluster. Set $z_i$ to $\tz$ if this is the case. 


\IncMargin{1em}
\begin{algorithm}[h]
  \caption{\protect\optimAdd: move center points from $\aleph$ to $\mcx$ if it decreases cost.
}\label{algo:optimadd}  
$\alephindex \leftarrow \{i \in [n];\, z_i = \aleph \}$\;
\If{$\alephindex !\!= \emptyset$}{
  $\supply \leftarrow \{\text{data points that are miserably matched to some $z_i$}\}$\;
  \For{$i$ \InSet \alephindex}{
    \lIf{$\supply = \emptyset$}{\textbf{break}}
    $\tz \leftarrow \sample(\supply,1)$\tcp*{draw uniformly at random from \supply}
    $\supply \leftarrow \supply \setminus \{\tz\}$\;
    $\newcluster, \newperm \leftarrow \optimizeCluster(\tz,\perm, i)$\;
    $\newhappypoints \leftarrow \{\text{points in \newcluster with a happy new match to $\tz$}\}$\;
    \lIf{$\newhappypoints\ !\!= \emptyset$}{
      $\tz \leftarrow \optimClusterCenter(\newhappypoints)$}
    $k_{\mcx} \leftarrow \#\{x \in \newcluster;\, x \in \mcx \}$\;
    $c_{\mathrm{new}} \leftarrow \sum_{x \in \newcluster} d'(x,\tz)^p$\;
    \If{$c_{\mathrm{new}} < k_{\mcx} * C^p$}{
      $z_i \leftarrow \tz$\;
      $\perm \leftarrow \newperm$\;
      $\supply \leftarrow \supply \setminus \{x \in \supply;\, x \text{ is happy}\}$\;
    }
  }
}
\KwRet{$\{ z_1, \ldots, z_n \}$}\;
\end{algorithm}\DecMargin{1em}

\subsection{An improved \texttt{kMeansBary} algorithm}
\label{ssec:improved}


For obtaining an algorithm with a reduced computational cost, we cut down on steps that are costly, but are not expected to influence the resulting local optimum in a decisive way. Since for now we treat the location problem at the cluster level (performed by \optimBary) as very general, allowing a wide range of metric spaces $(\mcx,d)$, we focus here on saving computations in the functions \optimPerm, \optimDelete, and \optimAdd. 

We have realized that by far the most additions and deletions of points take place in the first two iterations of the original algorithm (see also Figure~\ref{fig:evolution} below). Especially checking for addition of points is costly and after the first few iterations very rarely successful. Therefore we limit such checking henceforward to the first $N_{\mathrm{del/add}} = 5$ iteration steps. Some further heuristics could be applied in $\optimAdd$, but the gain in computation time is not so large and they can significantly change the outcome, which is why we decided against implementing them.

In $\optimPerm$ we cannot avoid doing matchings. However, the auction algorithm we use allows to solve a relaxation of the problem by stopping the $\eps$-scaling method early. In general, the auction algorithm with $\eps$-scaling based on a decreasing sequence $(\eps_1,\ldots,\eps_l)$ returns successively improved solutions that are guaranteed to lie within $n\hbit\eps_i$ of the optimal total cost after the $i$-th step, see~\citet[Proposition~1]{Bertsekas1988}. By representing rescaled distances as integers in $\{0,1,\ldots,10^9\}$, an optimal matching is obtained in the $l$-th step if $\eps_l < 1/n$. Our improved algorithm is based on the same $\eps$-vector as the original algorithm, which has components $\eps_i = \frac{1}{n+1} 10^{l-i}$, $1 \leq i \leq l$, where $l$ is chosen in such a way that $10^7 \leq \eps_1 < 10^8$. As a first improvement, we use the subsequence $(\eps_{a_{\it}}, \eps_{a_{\it}+1}, \ldots,\eps_{b_{\it}})$, where $a$ and $b$ are prespecified vectors of indices $\in \{1,2,\ldots,l\}$. A simple choice for $a$ and $b$ that tends to decrease the runtime noticeably is $a_{\it} = 1$ and $b_{\it} = \min\{\it,l\}$. Pseudocode for this is presented in Algorithm~\ref{algo:kmeansbary2}.

In practice we settled for a somewhat more sophisticated improvement. We choose $a = (1,1,1,3,3,3,\ldots,3,4)$ and $b = (1,2,3,4,6,8,\ldots,2\lfloor\frac{l-1}{2}\rfloor,l)$, and we use the sequence $(\eps_{a_{j}}, \eps_{a_{j}+1}, \ldots,\eps_{b_{j}})$, where $j=\it$ for $\it \in \{1,2,3\}$, and then $j$ is increased by $1$ each time the algorithm would otherwise converge or if the cost increases (which can only happen as long as the matchings are not optimal).

This strategy was chosen after analyzing the calculations of the algorithm with respect to the time each calculation takes. In the first two to three iterations there are a lot of changes in the positions of the barycenter points. Especially in the first iteration many points are deleted and added, which completely changes the assignments. Therefore we have to begin the assignment calculation with $\eps_1$ and to get more sensible results we get more precise with each of the first three iterations. 
After three iterations there are usually no big changes to the barycenter anymore, so we can reuse the assignment from the iteration before as a sensible starting solution and can omit $\eps_1$ and $\eps_2$ in return. Leaving out the first entries of $\eps$ too soon increases the runtime. Every time the algorithm converges, but has no guaranteed optimal assignment (i.e.\ $b_j < l$), $j$ is increased by $1$, meaning that the next two entries of $\eps$ are used too, until the end of $\eps$ is reached. 
Then we can safely leave out the first three entries of $\eps$ without increasing the runtime, because at this point the assignments from one iteration to the next only change very little.

\IncMargin{1em}
\begin{algorithm}[h]
  \caption{\FuncSty{kMeansBary2}. A simple version of the improved \FuncSty{kMeansBary}-algorithm.  Dependence on data \protect\pplist suppressed for simplicity.
}\label{algo:kmeansbary2}  
\input{$\center$, $\pplist$, $\delta > 0$, $N$ are as in the original \FuncSty{kMeansBary}-algorithm;\\
       $N_{\mathrm{del/add}}$ number of iterations during which we perform delete/add steps.}
\output{Locally optimal pseudo-barycenter $\center$.}
\BlankLine
$l \leftarrow \bigl\lceil \log_{10}(10^8 / \frac{1}{n+1}) \bigr\rceil$\;
$\epsvec \leftarrow \bigl( \frac{1}{n+1} 10^{l-i} \bigr)_{1 \leq i \leq l}$\;
$\perm, \cost \leftarrow \optimPerm(\center,\epsvec)$\;
\For{$\it \leftarrow 1$ \KwTo $N$}{
  $\costold \leftarrow \cost$\;
  $\center \leftarrow \optimBary(\perm, \center)$\;
\If{$\it \leq N_{\mathrm{del/add}}$}{
  $\center \leftarrow \optimDelete(\perm, \center)$\;
  $\center \leftarrow \optimAdd(\perm, \center)$\;
}
\eIf{$\it < l$}{
	$\perm, \cost \leftarrow \optimPerm(\center,\epsvec[1:\it])$\;
}{
  $\perm, \cost \leftarrow \optimPerm(\center,\epsvec)$\;
  \lIf(\tcp*[f]{difference always nonnegative}){$\costold-\cost < \delta$}{
    \textbf{break}}
}
}
\KwRet{$\center$}\tcp*{and warn if the loop has run out}
\end{algorithm}\DecMargin{1em}

\subsection{Practical aspects}

As it turns out the upper bound $\tn$ on the cardinality of the barycenter from Proposition~\ref{prop:mdagengen} is often far too large in practice. For efficiency reasons we typically run the algorithm with a number $n \geq \max\{n_j;\, 1 \leq j \leq k\}$ that is much smaller than $\tn$. We generate a starting point pattern \center by picking $\frac{1}{k} \sum_{j=1}^k \abs{\xi_j}$ points uniformly at random from the underlying observation window. In a first step all point patterns are filled to $n$ points by adding points at $\aleph$. Then Algorithm~\ref{algo:kmeansbary} or~\ref{algo:kmeansbary2} is run.

Figure~\ref{fig:evolution} shows a typical run of Algorithm~\ref{algo:kmeansbary} in the case of an i.i.d.\ sample $\xi_1, \ldots, \xi_k$ of point patterns in $\RR^2$ generated from a similar distribution as studied in Section~\ref{sec:simulstudy}. We use Euclidean distance and $p=2$. The current barycenter is marked by blue points. Typically the random starting point pattern is not a good approximation to the resulting barycenter. Therefore many points are deleted in the first iteration. Many other ones are added at or moved to more cost efficient spots. Regardless of the starting pattern the algorithm typically attains a reasonably looking configuration after a single iteration. After that hardly any points are added or deleted any more. The algorithm mostly moves a few individual barycenter points around each time.
\begin{figure}[h]
  \includegraphics[scale=0.55]{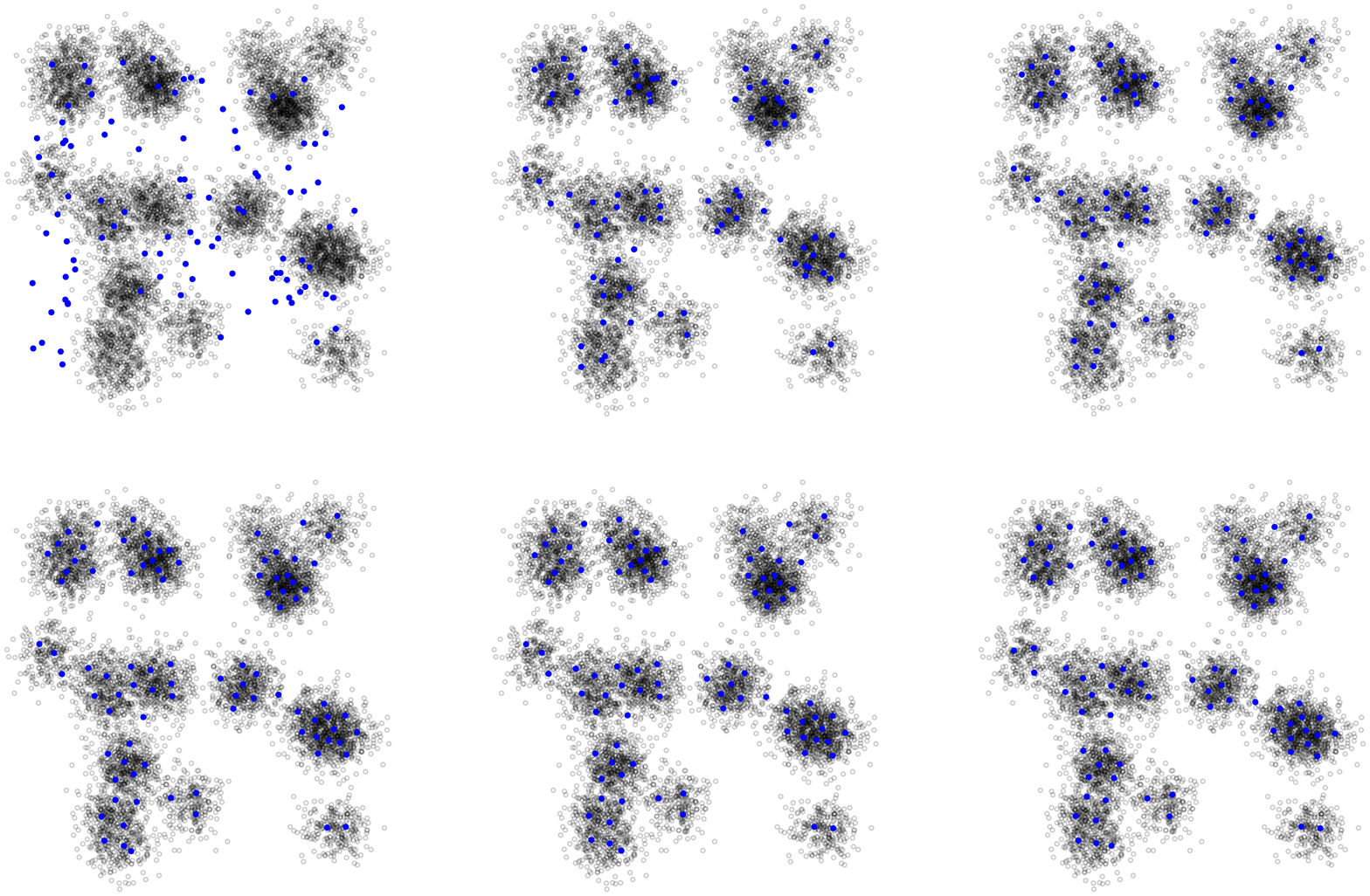}
    \vspace*{-14mm}

\caption{Stepwise evolution of the barycenter for $k=80$, $m_{\#} = 100$. In the first iteration $32$ points are deleted and $25$ added. After that only movements take place.}
\label{fig:evolution}
\end{figure}

\section{Simulation study}
\label{sec:simulstudy}

In this section we present a simulation study for evaluating the algorithms described in Section~\ref{sec:algorithms} for point patterns in $\RR^2$ using squared Euclidean cost. Unfortunately it is not feasible for larger data examples to compute the actual barycenter as a ground truth. Even when treating the point patterns as empirical measures and solving the simpler and much better studied problem of finding a barycenter in the space of probability measures (with respect to the Wasserstein metric $W_2$), the computation times for problems close to our smallest examples below range from minutes to hours; see \cite{AnderesEtAl2016} and \cite{BorgwardtPatterson2018}.
As a replacement of comparing with the actual barycenter, we assess the range of the final objective function values. In addition we evaluate the time performance of the default algorithm, and compare both objective function values and timings to the improved algorithm.

As problem instances we created sets of $k$ point patterns in $\RR^2$ having mean cardinality of $m_{\#}$ in each pattern. The cardinalities $n_j$, $j \in [k]$, of the individual point patterns were determined by one of the following methods:
\begin{tightenumerate}
  \item[(i)] by setting $n_j = m_{\#}$ \emph{(deterministic cardinality)}
  \item[(ii)] by sampling $n_j$ from a binomial distribution with mean $m_{\#}$ and variance $\approx 1$\\ \emph{(low-variance cardinality)}
  \item[(iii)] by sampling $n_j$ from a Poisson distribution with parameter $m_{\#}$\\  \emph{(high-variance cardinality)} 
  \end{tightenumerate}

The points were distributed according to a balanced mixture of $N \in \{5,10,15\}$ rotationally symmetric normal densities centered at fixed locations in $[0,1]^2$ and having standard deviation $\sigma \in \{0.05,0.1,0.2\}$. Figure~\ref{fig:centerscenarios} gives examples under the three center scenarios for $k=20$, deterministic cardinality $n_j = m_{\#} = 20$ and $\sigma=0.05$.
\begin{figure}[h]
  \hspace*{-8mm}\includegraphics[width=0.4\textwidth]{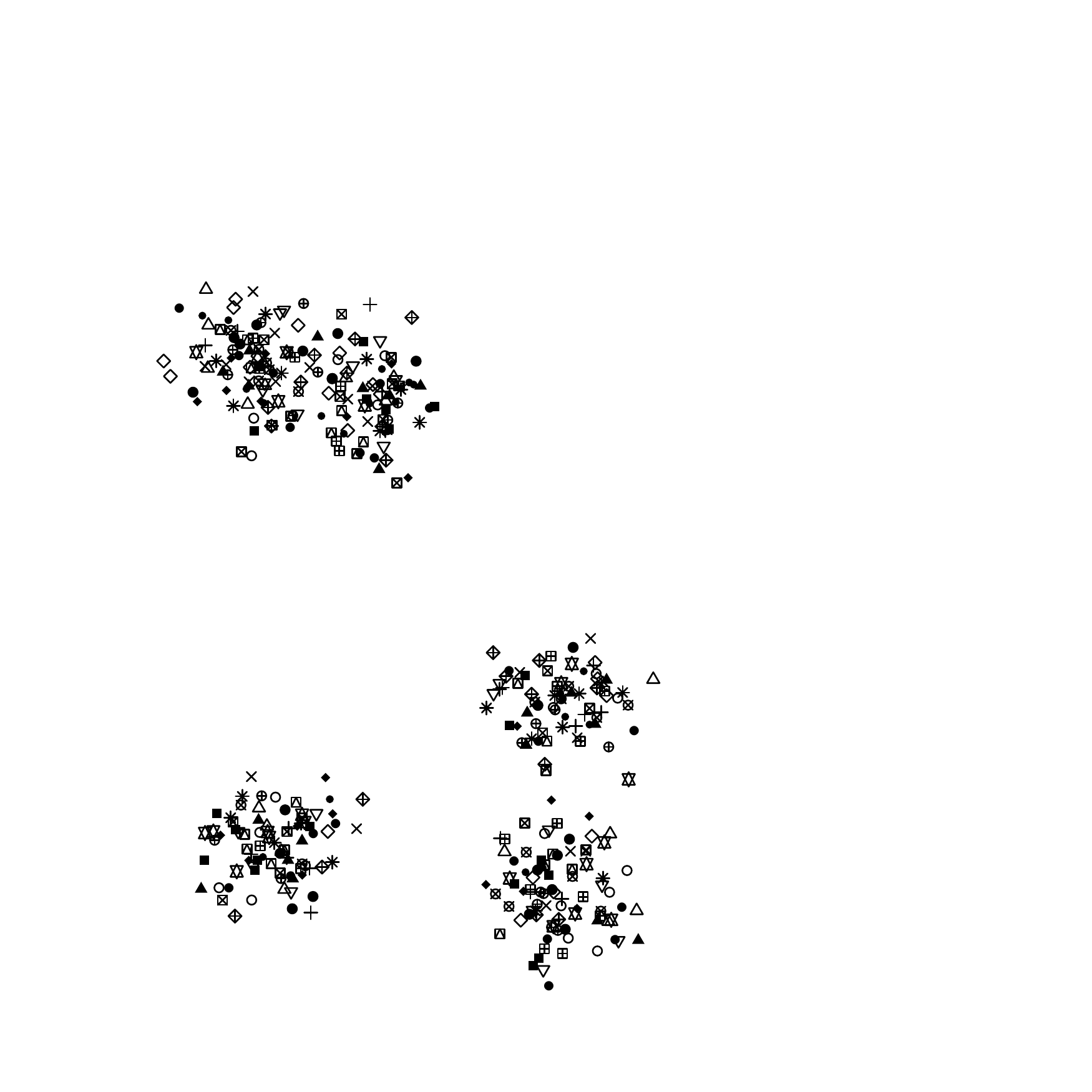}\hspace*{-18mm}
  \includegraphics[width=0.4\textwidth]{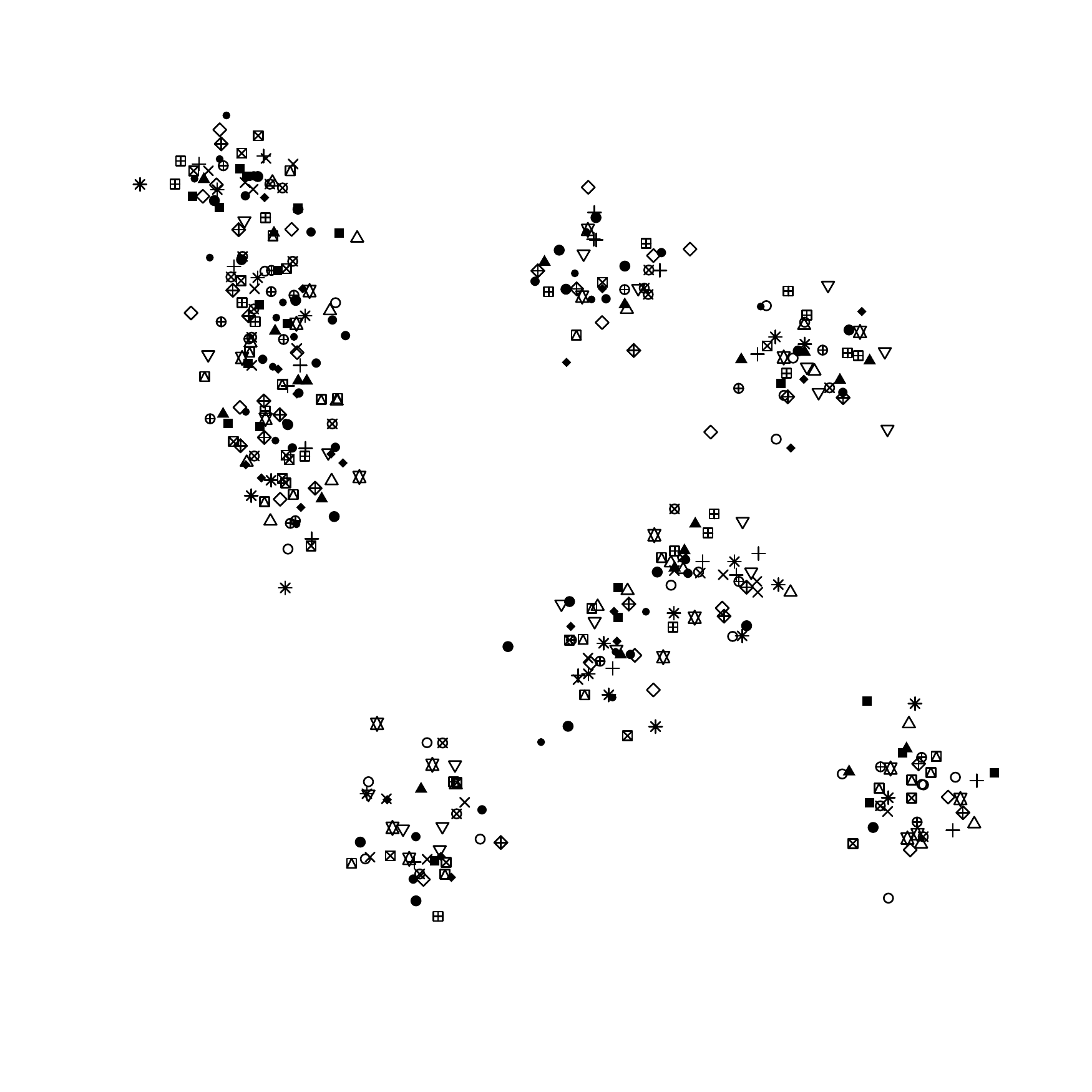}\hspace*{-4mm}
  \includegraphics[width=0.4\textwidth]{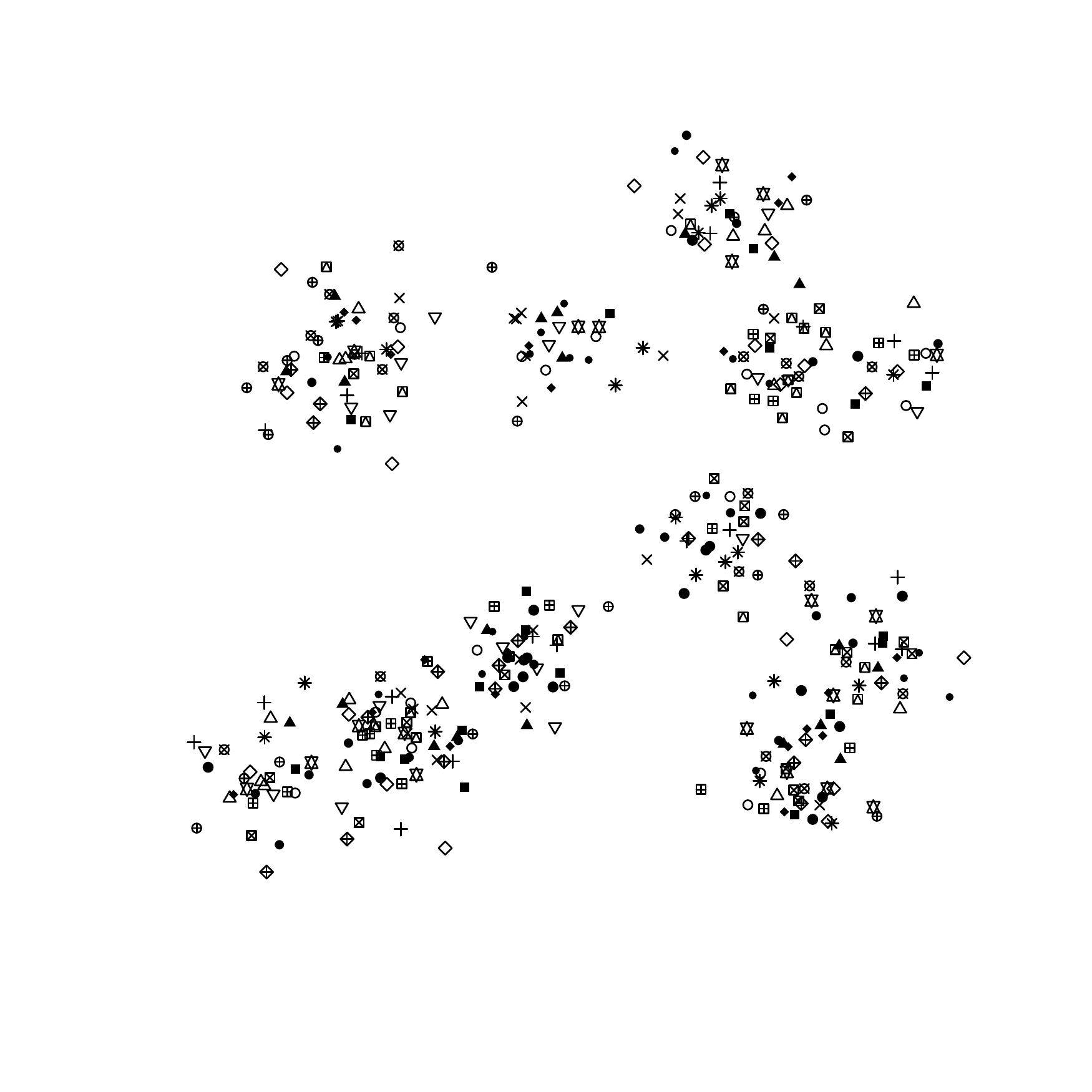}\hspace*{-4mm}\\[-12mm] \hspace*{2mm}
\caption{20 point patterns with 20 points each from the three different center scenarios $N=5, 10, 15$ for $\sigma=0.05$.}
\label{fig:centerscenarios} 
\end{figure}

We chose five $(k,m_{\#})$ pairs $(20,20), (20,50), (50,20), (50,50)$ and $(100,100)$, which in combination with $N$ varying in $\{5,10,15\}$, $\sigma$ in $\{0.05, 0.1, 0.2\}$ and the three cardinality distributions yield a total of $5 \times 3^3 = 135$ scenarios. We created $100$ instances for each scenario.

Our algorithms from Section~\ref{sec:algorithms} were run from ten starting solutions whose cardinalities matched the mean number of data points and whose points were sampled uniformly at random from $[0,1]^2$. In a pilot experiment this tended to give somewhat better local minima than starting from a random sample of all data points combined. The starting point patterns were independently chosen for each instance, but the same for both algorithms.

We first consider the original algorithm presented in Section~\ref{sec:algorithms}. 
Table~\ref{tab:distvanilla} gives the maximum relative deviation from the minimum $d_{\min}$ of the resulting objective function values among the ten starting solutions, i.e.\ $\frac{d_{\max}-d_{\min}}{d_{\min}}$. We can see that the maximal objective function value among the ten runs rarely exceeds the minimum value by more than 5\%. This percentage is rather higher for the deterministic and low-variance cardinalities and when clusters in the (unmarked) superposition of the point patterns are well separated (small $N$ and $\sigma$). This may well be explicable by the fact that typically many pairs can be matched over short distances in these situations such that wrong clustering decisions come typically at a higher relative cost. Figure~\ref{fig:Nsigma} supports this by showing that the total objective function values within each problem size are lower for well separated clusters.

A further smaller experiment following up on the scenarios that exhibited the poorest performance for ten starting patterns showed that the margin of 5\% increases to 8\% when basing the maximum relative deviation from the minimum on 100 starting patterns.

\begin{figure}[h]
  \includegraphics[height=0.37\linewidth, width=1.02\linewidth]{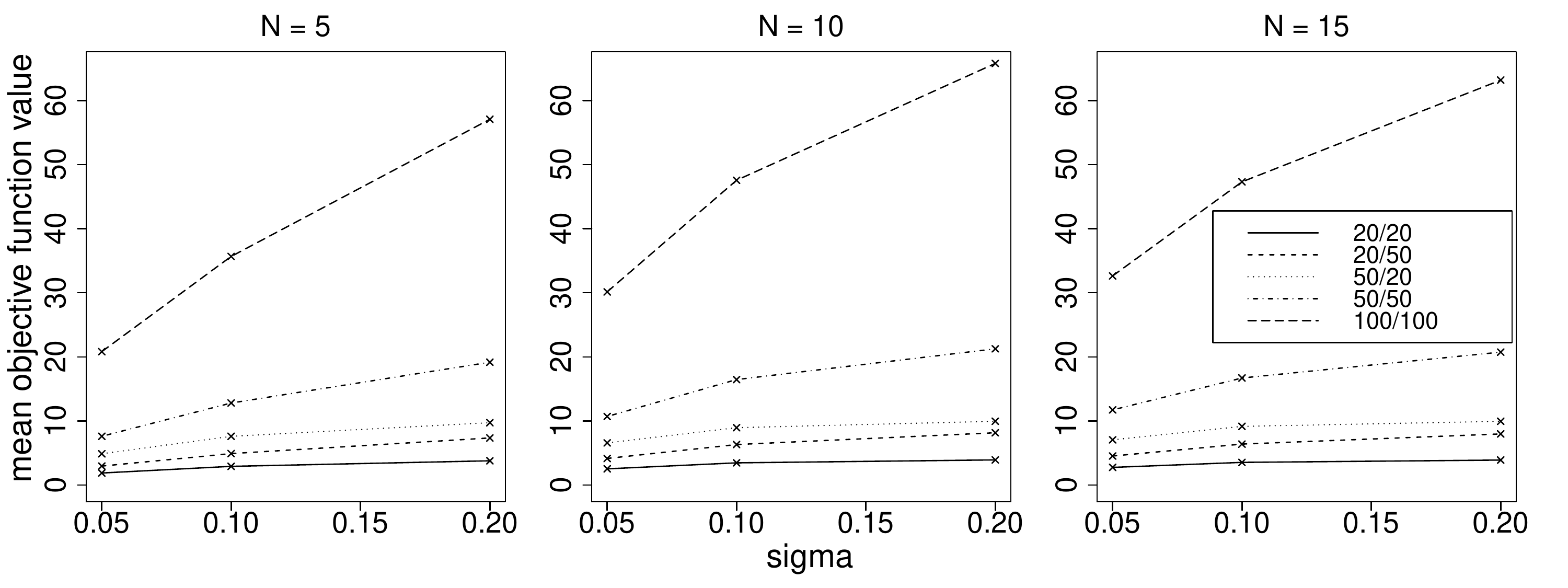}\hspace*{-7mm}
\vspace*{-4mm}
  
\caption{Mean objective function values over all instances as function of $\sigma$ for different $N$.}
\label{fig:Nsigma} 
\end{figure}

For the improved algorithm from Subsection~\ref{ssec:improved} we compute the maximum relative deviation of its objective function values from the minimum $d_{\min}$ of the corresponding values of the original algorithm, i.e.\  
$\frac{d^{*}_{\max}-d_{\min}}{d_{\min}}$, where $d^{*}_{\max}$ is the maximum of the objective function values of the improved algorithm. As seen in Table~\ref{tab:distpimpedvsvanilla} the performance is no worse than for the original algorithm in spite of the reduced amount of computations performed.

\begin{table}[h]
  \caption{\emph{Original algorithm.} Maximum relative deviation from the minimum objective function value among ten starting solutions. Means taken over 100 instances, with $5\%$- and $95\%$-quantile in parentheses. The first block of $9$ rows corresponds to the deterministic cardinality, the second block to the low-variance cardinality and the last block to the high-variance cardinality.}
  \label{tab:distvanilla}
  \vspace*{3mm}
  
\footnotesize
\centering
\begin{tabular}{|c|c|r|r|r|r|r|}
  \hline
 N & $\sigma$ & 20/20 & 20/50 & 50/20 & 50/50 & 100/100 \\ 
\hline
  & 0.05 & 0.05 (0.03, 0.07) & 0.04 (0.02, 0.06) & 0.04 (0.03, 0.06) & 0.04 (0.02, 0.05) & 0.03 (0.02, 0.05) \\ 
\cline{2-7}
5 & 0.1 & 0.04 (0.02, 0.06) & 0.03 (0.02, 0.05) & 0.03 (0.02, 0.04) & 0.03 (0.02, 0.04) & 0.03 (0.02, 0.05) \\ 
\cline{2-7}
 & 0.2 & 0.02 (0.01, 0.03) & 0.02 (0.01, 0.03) & 0.01 (0.01, 0.02) & 0.02 (0.01, 0.02) & 0.01 (0.01, 0.02) \\ 
\cline{1-7}
 & 0.05 & 0.04 (0.02, 0.06) & 0.04 (0.02, 0.06) & 0.03 (0.02, 0.05) & 0.03 (0.02, 0.05) & 0.03 (0.02, 0.05) \\ 
\cline{2-7}
10 & 0.1 & 0.02 (0.02, 0.04) & 0.03 (0.02, 0.04) & 0.02 (0.01, 0.03) & 0.02 (0.01, 0.03) & 0.02 (0.01, 0.03) \\ 
\cline{2-7}
 & 0.2 & 0.01 (0.00, 0.02) & 0.02 (0.01, 0.03) & 0.00 (0.00, 0.01) & 0.01 (0.01, 0.02) & 0.01 (0.01, 0.01) \\ 
\cline{1-7}
 & 0.05 & 0.03 (0.02, 0.05) & 0.03 (0.02, 0.05) & 0.03 (0.01, 0.04) & 0.03 (0.02, 0.04) & 0.03 (0.02, 0.04) \\ 
\cline{2-7}
15 & 0.1 & 0.02 (0.01, 0.04) & 0.03 (0.02, 0.04) & 0.02 (0.01, 0.02) & 0.02 (0.01, 0.03) & 0.02 (0.01, 0.02) \\ 
\cline{2-7}
 & 0.2 & 0.02 (0.01, 0.03) & 0.02 (0.01, 0.03) & 0.01 (0.00, 0.01) & 0.01 (0.01, 0.02) & 0.01 (0.01, 0.01) \\ 
\hline \hline
 & 0.05 & 0.05 (0.03, 0.07) & 0.04 (0.02, 0.05) & 0.04 (0.02, 0.06) & 0.04 (0.02, 0.05) & 0.03 (0.02, 0.05) \\ 
\cline{2-7}
5 & 0.1 & 0.04 (0.02, 0.05) & 0.03 (0.02, 0.05) & 0.03 (0.01, 0.04) & 0.03 (0.02, 0.04) & 0.03 (0.02, 0.05) \\ 
\cline{2-7}
 & 0.2 & 0.02 (0.01, 0.03) & 0.02 (0.01, 0.04) & 0.01 (0.01, 0.02) & 0.02 (0.01, 0.02) & 0.01 (0.01, 0.02) \\ 
\cline{1-7}
 & 0.05 & 0.03 (0.02, 0.06) & 0.04 (0.02, 0.06) & 0.03 (0.02, 0.04) & 0.03 (0.02, 0.05) & 0.04 (0.02, 0.05) \\ 
\cline{2-7}
10 & 0.1 & 0.02 (0.01, 0.04) & 0.03 (0.02, 0.04) & 0.02 (0.01, 0.03) & 0.02 (0.01, 0.03) & 0.02 (0.01, 0.03) \\ 
\cline{2-7}
 & 0.2 & 0.01 (0.01, 0.02) & 0.02 (0.01, 0.03) & 0.01 (0.00, 0.01) & 0.01 (0.01, 0.02) & 0.01 (0.01, 0.01) \\ 
\cline{1-7}
 & 0.05 & 0.03 (0.02, 0.05) & 0.03 (0.02, 0.05) & 0.02 (0.01, 0.04) & 0.03 (0.02, 0.04) & 0.03 (0.01, 0.04) \\ 
\cline{2-7}
15 & 0.1 & 0.02 (0.01, 0.04) & 0.03 (0.02, 0.04) & 0.01 (0.01, 0.02) & 0.02 (0.01, 0.03) & 0.02 (0.01, 0.02) \\ 
\cline{2-7}
 & 0.2 & 0.01 (0.01, 0.02) & 0.02 (0.01, 0.03) & 0.01 (0.00, 0.01) & 0.02 (0.01, 0.02) & 0.01 (0.01, 0.01) \\ 
\hline \hline
  & 0.05 & 0.03 (0.02, 0.05) & 0.02 (0.01, 0.04) & 0.03 (0.01, 0.04) & 0.02 (0.01, 0.03) & 0.01 (0.01, 0.02) \\ 
\cline{2-7}
5 & 0.1 & 0.03 (0.02, 0.05) & 0.03 (0.02, 0.04) & 0.02 (0.01, 0.04) & 0.02 (0.01, 0.03) & 0.02 (0.01, 0.03) \\ 
\cline{2-7}
 & 0.2 & 0.02 (0.01, 0.03) & 0.02 (0.01, 0.03) & 0.01 (0.01, 0.01) & 0.02 (0.01, 0.02) & 0.01 (0.01, 0.02) \\ 
\cline{1-7}
 & 0.05 & 0.03 (0.02, 0.04) & 0.03 (0.01, 0.04) & 0.02 (0.01, 0.04) & 0.02 (0.01, 0.04) & 0.02 (0.01, 0.03) \\ 
\cline{2-7}
10 & 0.1 & 0.02 (0.01, 0.04) & 0.02 (0.01, 0.03) & 0.01 (0.01, 0.02) & 0.02 (0.01, 0.03) & 0.01 (0.01, 0.02) \\ 
\cline{2-7}
 & 0.2 & 0.01 (0.00, 0.02) & 0.02 (0.01, 0.03) & 0.00 (0.00, 0.01) & 0.01 (0.01, 0.02) & 0.01 (0.01, 0.01) \\ 
\cline{1-7}
 & 0.05 & 0.03 (0.01, 0.04) & 0.03 (0.02, 0.04) & 0.02 (0.01, 0.04) & 0.02 (0.01, 0.03) & 0.01 (0.01, 0.02) \\ 
\cline{2-7}
15 & 0.1 & 0.02 (0.01, 0.03) & 0.03 (0.01, 0.04) & 0.01 (0.01, 0.02) & 0.02 (0.01, 0.03) & 0.01 (0.01, 0.02) \\ 
\cline{2-7}
 & 0.2 & 0.01 (0.00, 0.02) & 0.02 (0.01, 0.03) & 0.00 (0.00, 0.01) & 0.01 (0.01, 0.02) & 0.01 (0.01, 0.02) \\ 
\hline
\end{tabular}
\end{table}
\clearpage

\begin{table}
  \caption{\emph{Improved algorithm.} Maximum relative deviation from the minimum objective function value of the original algorithm (both based on the same ten starting solutions). Means over 100 instances, with $5\%$- and $95\%$-quantile in parentheses. The first block of $9$ rows corresponds to the deterministic cardinality, the second block to the low-variance cardinality and the last block to the high-variance cardinality.}
  \label{tab:distpimpedvsvanilla}
\vspace*{3mm}
  
  \footnotesize
\centering
\begin{tabular}{|c|c|r|r|r|r|r|}
  \hline
 N & $\sigma$ & 20/20 & 20/50 & 50/20 & 50/50 & 100/100 \\ 
\hline
 & 0.05 & 0.04 (0.02, 0.06) & 0.04 (0.02, 0.06) & 0.04 (0.02, 0.06) & 0.04 (0.02, 0.06) & 0.03 (0.02, 0.05) \\ 
\cline{2-7}
5 & 0.1 & 0.04 (0.02, 0.06) & 0.03 (0.02, 0.05) & 0.03 (0.02, 0.04) & 0.03 (0.02, 0.04) & 0.03 (0.02, 0.05) \\ 
\cline{2-7}
 & 0.2 & 0.02 (0.01, 0.03) & 0.02 (0.02, 0.03) & 0.01 (0.00, 0.02) & 0.02 (0.01, 0.02) & 0.01 (0.01, 0.02) \\ 
\cline{1-7}
 & 0.05 & 0.04 (0.02, 0.06) & 0.04 (0.02, 0.05) & 0.03 (0.02, 0.05) & 0.04 (0.02, 0.05) & 0.03 (0.02, 0.05) \\ 
\cline{2-7}
10 & 0.1 & 0.02 (0.01, 0.03) & 0.03 (0.02, 0.04) & 0.02 (0.01, 0.02) & 0.02 (0.01, 0.03) & 0.02 (0.01, 0.03) \\ 
\cline{2-7}
 & 0.2 & 0.01 (0.00, 0.02) & 0.02 (0.01, 0.03) & 0.00 (0.00, 0.01) & 0.02 (0.01, 0.02) & 0.01 (0.01, 0.01) \\ 
\cline{1-7}
 & 0.05 & 0.03 (0.02, 0.05) & 0.03 (0.02, 0.05) & 0.03 (0.01, 0.04) & 0.03 (0.02, 0.04) & 0.03 (0.02, 0.04) \\ 
\cline{2-7}
15 & 0.1 & 0.02 (0.01, 0.04) & 0.03 (0.02, 0.04) & 0.02 (0.01, 0.02) & 0.02 (0.01, 0.03) & 0.02 (0.01, 0.02) \\ 
\cline{2-7}
 & 0.2 & 0.01 (0.00, 0.02) & 0.02 (0.01, 0.03) & 0.00 (0.00, 0.01) & 0.02 (0.01, 0.02) & 0.01 (0.01, 0.01) \\
\hline \hline
 & 0.05 & 0.04 (0.02, 0.07) & 0.04 (0.02, 0.05) & 0.04 (0.02, 0.06) & 0.03 (0.02, 0.05) & 0.03 (0.02, 0.05) \\ 
\cline{2-7}
5 & 0.1 & 0.03 (0.02, 0.05) & 0.03 (0.02, 0.05) & 0.03 (0.01, 0.04) & 0.03 (0.02, 0.04) & 0.03 (0.02, 0.04) \\ 
\cline{2-7}
 & 0.2 & 0.02 (0.01, 0.03) & 0.03 (0.01, 0.04) & 0.01 (0.00, 0.02) & 0.02 (0.01, 0.03) & 0.01 (0.01, 0.02) \\ 
\cline{1-7}
 & 0.05 & 0.03 (0.02, 0.05) & 0.03 (0.02, 0.06) & 0.03 (0.01, 0.04) & 0.03 (0.02, 0.05) & 0.04 (0.02, 0.05) \\ 
\cline{2-7}
10 & 0.1 & 0.02 (0.01, 0.04) & 0.03 (0.01, 0.04) & 0.02 (0.01, 0.02) & 0.02 (0.01, 0.03) & 0.02 (0.01, 0.03) \\ 
\cline{2-7}
 & 0.2 & 0.01 (0.00, 0.02) & 0.02 (0.01, 0.03) & 0.00 (0.00, 0.01) & 0.01 (0.01, 0.02) & 0.01 (0.01, 0.01) \\ 
\cline{1-7}
 & 0.05 & 0.03 (0.02, 0.05) & 0.03 (0.02, 0.04) & 0.02 (0.01, 0.04) & 0.03 (0.02, 0.04) & 0.02 (0.01, 0.04) \\ 
\cline{2-7}
15 & 0.1 & 0.02 (0.01, 0.03) & 0.03 (0.02, 0.04) & 0.01 (0.01, 0.02) & 0.02 (0.01, 0.03) & 0.01 (0.01, 0.02) \\ 
\cline{2-7}
 & 0.2 & 0.01 (0.00, 0.02) & 0.02 (0.02, 0.03) & 0.00 (0.00, 0.01) & 0.02 (0.01, 0.02) & 0.01 (0.00, 0.01) \\
\hline \hline
 & 0.05 & 0.03 (0.02, 0.05) & 0.02 (0.01, 0.03) & 0.03 (0.01, 0.04) & 0.02 (0.01, 0.03) & 0.01 (0.00, 0.02) \\ 
\cline{2-7}
5 & 0.1 & 0.03 (0.02, 0.04) & 0.03 (0.01, 0.04) & 0.02 (0.01, 0.04) & 0.02 (0.01, 0.03) & 0.02 (0.01, 0.03) \\ 
\cline{2-7}
 & 0.2 & 0.02 (0.01, 0.03) & 0.02 (0.01, 0.03) & 0.01 (0.00, 0.01) & 0.02 (0.01, 0.02) & 0.01 (0.01, 0.02) \\ 
\cline{1-7}
 & 0.05 & 0.03 (0.01, 0.04) & 0.03 (0.01, 0.04) & 0.02 (0.01, 0.04) & 0.02 (0.01, 0.03) & 0.02 (0.01, 0.03) \\ 
\cline{2-7}
10 & 0.1 & 0.02 (0.01, 0.03) & 0.02 (0.02, 0.04) & 0.01 (0.01, 0.02) & 0.02 (0.01, 0.03) & 0.01 (0.01, 0.02) \\ 
\cline{2-7}
 & 0.2 & 0.01 (0.00, 0.02) & 0.02 (0.01, 0.03) & 0.00 (0.00, 0.01) & 0.01 (0.01, 0.02) & 0.01 (0.01, 0.01) \\ 
\cline{1-7}
 & 0.05 & 0.02 (0.01, 0.04) & 0.02 (0.02, 0.04) & 0.02 (0.01, 0.04) & 0.02 (0.01, 0.03) & 0.01 (0.01, 0.02) \\ 
\cline{2-7}
15 & 0.1 & 0.02 (0.01, 0.03) & 0.02 (0.01, 0.04) & 0.01 (0.01, 0.02) & 0.02 (0.01, 0.03) & 0.01 (0.01, 0.02) \\ 
\cline{2-7}
 & 0.2 & 0.01 (0.00, 0.01) & 0.02 (0.01, 0.03) & 0.00 (0.00, 0.01) & 0.01 (0.01, 0.02) & 0.01 (0.01, 0.02) \\
\hline
\end{tabular}
\end{table}
\clearpage

We finally turn to the computation times. We present the total runtimes in seconds for the ten runs with different starting patterns. This corresponds to the realistic situation of selecting as (pseudo-)barycenter the solution with the smallest local minimum in ten runs. It also provides some more stability for the means and quantiles given in Tables~\ref{tab:timebof10} and~\ref{tab:timebof10pimped}.

Table~\ref{tab:timebof10} gives the runtimes for the original algorithm. We see that individual runs of as large scenarios as 100 patterns with 100 points on average only take a few seconds.

From Table~\ref{tab:timebof10pimped} we see that the runtimes for the improved algorithm are even considerably lower, and for some of the larger problems they have less than half of the original runtimes (at virtually no loss with regard to the objective function value as we have seen before). It is to be expected that this ratio becomes even smaller if the problem size is further increased.

\begin{table}
  \caption{\emph{Original algorithm.} Time in seconds for a total of ten runs with random starting patterns. Means over 100 instances, with $5\%$- and $95\%$-quantile in parentheses. The first block of $9$ rows corresponds to the deterministic cardinality, the second block to the low-variance cardinality and the last block to the high-variance cardinality.}
  \label{tab:timebof10}
\vspace*{3mm}
  
\footnotesize
\centering
\begin{tabular}{|c|c|r|r|r|r|r|}
  \hline
 N & $\sigma$ & 20/20 & 20/50 & 50/20 & 50/50 & 100/100 \\ 
\hline
 & 0.05 & 0.48 (0.47, 0.50) & 0.89 (0.85, 0.93) & 1.02 (0.99, 1.05) & 2.37 (2.24, 2.52) & 20.79 (19.42, 22.32) \\ 
\cline{2-7}
5 & 0.1 & 0.52 (0.51, 0.54) & 1.05 (0.99, 1.10) & 1.14 (1.10, 1.19) & 3.13 (2.92, 3.36) & 30.26 (27.51, 32.92) \\ 
\cline{2-7}
 & 0.2 & 0.51 (0.49, 0.53) & 1.14 (1.07, 1.23) & 1.05 (0.99, 1.12) & 3.55 (3.26, 3.96) & 37.95 (34.70, 42.22) \\ 
\cline{1-7}
 & 0.05 & 0.49 (0.48, 0.51) & 0.93 (0.89, 0.97) & 1.05 (1.01, 1.08) & 2.49 (2.37, 2.64) & 22.15 (20.33, 24.63) \\ 
\cline{2-7}
10 & 0.1 & 0.52 (0.51, 0.54) & 1.08 (1.03, 1.15) & 1.14 (1.09, 1.20) & 3.18 (2.96, 3.47) & 33.83 (30.03, 37.89) \\ 
\cline{2-7}
 & 0.2 & 0.48 (0.45, 0.50) & 1.16 (1.10, 1.24) & 0.91 (0.85, 0.97) & 3.56 (3.27, 3.91) & 40.30 (36.14, 44.69) \\ 
\cline{1-7}
 & 0.05 & 0.50 (0.49, 0.51) & 0.95 (0.91, 1.00) & 1.06 (1.03, 1.09) & 2.60 (2.44, 2.80) & 24.25 (22.51, 25.93) \\ 
\cline{2-7}
15 & 0.1 & 0.52 (0.50, 0.53) & 1.08 (1.02, 1.15) & 1.13 (1.08, 1.17) & 3.20 (2.93, 3.53) & 32.70 (29.62, 35.91) \\ 
\cline{2-7}
 & 0.2 & 0.49 (0.47, 0.51) & 1.14 (1.08, 1.22) & 0.91 (0.85, 0.99) & 3.51 (3.12, 3.85) & 37.77 (33.86, 42.21) \\ 
\hline \hline
 & 0.05 & 0.51 (0.49, 0.53) & 0.92 (0.89, 0.96) & 1.05 (1.02, 1.08) & 2.51 (2.38, 2.67) & 21.43 (19.89, 23.47) \\ 
\cline{2-7}
5 & 0.1 & 0.53 (0.52, 0.54) & 1.08 (1.02, 1.13) & 1.17 (1.12, 1.22) & 3.20 (3.00, 3.44) & 30.93 (28.03, 33.42) \\ 
\cline{2-7}
 & 0.2 & 0.54 (0.51, 0.57) & 1.18 (1.10, 1.25) & 1.08 (1.00, 1.15) & 3.71 (3.34, 4.06) & 38.97 (34.69, 42.74) \\ 
\cline{1-7}
 & 0.05 & 0.51 (0.49, 0.52) & 0.96 (0.91, 1.00) & 1.07 (1.04, 1.11) & 2.57 (2.42, 2.73) & 22.70 (20.97, 24.87) \\ 
\cline{2-7}
10 & 0.1 & 0.54 (0.52, 0.56) & 1.12 (1.06, 1.20) & 1.17 (1.12, 1.23) & 3.36 (3.07, 3.68) & 34.18 (30.47, 37.91) \\ 
\cline{2-7}
 & 0.2 & 0.49 (0.46, 0.51) & 1.19 (1.12, 1.27) & 0.93 (0.87, 1.00) & 3.72 (3.35, 4.11) & 41.21 (35.90, 46.22) \\ 
\cline{1-7}
 & 0.05 & 0.51 (0.50, 0.53) & 0.98 (0.94, 1.03) & 1.09 (1.06, 1.13) & 2.70 (2.51, 2.93) & 25.07 (22.73, 27.86) \\ 
\cline{2-7}
15 & 0.1 & 0.53 (0.52, 0.55) & 1.11 (1.05, 1.19) & 1.17 (1.12, 1.22) & 3.34 (3.11, 3.62) & 33.94 (30.91, 37.06) \\ 
\cline{2-7}
 & 0.2 & 0.50 (0.46, 0.52) & 1.19 (1.13, 1.26) & 0.93 (0.87, 1.03) & 3.66 (3.35, 3.95) & 38.60 (35.05, 42.61) \\ 
\hline \hline
 & 0.05 & 0.58 (0.53, 0.64) & 1.27 (1.11, 1.53) & 1.37 (1.24, 1.55) & 4.12 (3.38, 4.97) & 39.66 (33.33, 46.96) \\ 
\cline{2-7}
5 & 0.1 & 0.62 (0.57, 0.69) & 1.46 (1.26, 1.74) & 1.57 (1.36, 1.83) & 5.04 (4.05, 6.12) & 55.56 (45.67, 65.94) \\ 
\cline{2-7}
 & 0.2 & 0.59 (0.53, 0.67) & 1.58 (1.37, 1.88) & 1.30 (1.11, 1.47) & 5.62 (4.64, 6.86) & 65.86 (54.51, 77.49) \\ 
\cline{1-7}
 & 0.05 & 0.59 (0.54, 0.65) & 1.27 (1.09, 1.51) & 1.36 (1.24, 1.55) & 4.10 (3.38, 4.99) & 41.13 (35.65, 47.79) \\ 
\cline{2-7}
10 & 0.1 & 0.60 (0.55, 0.67) & 1.48 (1.29, 1.79) & 1.48 (1.29, 1.68) & 5.14 (4.34, 6.14) & 60.16 (50.16, 70.58) \\ 
\cline{2-7}
 & 0.2 & 0.54 (0.49, 0.61) & 1.60 (1.37, 1.87) & 1.06 (0.94, 1.23) & 5.66 (4.62, 6.80) & 69.19 (58.12, 83.04) \\ 
\cline{1-7}
 & 0.05 & 0.59 (0.55, 0.65) & 1.32 (1.16, 1.56) & 1.41 (1.26, 1.64) & 4.31 (3.59, 5.22) & 45.37 (38.88, 51.93) \\ 
\cline{2-7}
15 & 0.1 & 0.61 (0.55, 0.70) & 1.46 (1.24, 1.69) & 1.47 (1.31, 1.66) & 5.27 (4.42, 6.41) & 60.57 (50.30, 70.68) \\ 
\cline{2-7}
 & 0.2 & 0.55 (0.49, 0.61) & 1.55 (1.37, 1.81) & 1.06 (0.94, 1.23) & 5.73 (4.72, 7.23) & 66.25 (56.73, 79.16) \\ 
\hline
\end{tabular}
\end{table}

\begin{table}
  \caption{\emph{Improved algorithm.} Time in seconds for a total of ten runs with random starting patterns. Means over 100 instances, with $5\%$- and $95\%$-quantile in parentheses. The first block of $9$ rows corresponds to the deterministic cardinality, the second block to the low-variance cardinality and the last block to the high-variance cardinality.}
  \label{tab:timebof10pimped}
\vspace*{3mm}
  
\footnotesize
\centering
\begin{tabular}{|c|c|r|r|r|r|r|}
  \hline
 N & $\sigma$ & 20/20 & 20/50 & 50/20 & 50/50 & 100/100 \\ 
\hline
 & 0.05 & 0.48 (0.47, 0.48) & 0.83 (0.81, 0.85) & 0.97 (0.96, 0.99) & 2.05 (1.99, 2.12) & 13.68 (13.02, 14.57) \\ 
\cline{2-7}
5 & 0.1 & 0.50 (0.49, 0.50) & 0.92 (0.89, 0.96) & 1.03 (1.01, 1.05) & 2.42 (2.29, 2.60) & 16.51 (15.38, 17.85) \\ 
\cline{2-7}
 & 0.2 & 0.50 (0.49, 0.50) & 0.93 (0.89, 0.98) & 1.00 (0.98, 1.02) & 2.47 (2.33, 2.65) & 18.25 (16.87, 20.57) \\ 
\cline{1-7}
 & 0.05 & 0.48 (0.48, 0.49) & 0.84 (0.82, 0.87) & 0.98 (0.97, 1.00) & 2.06 (1.97, 2.14) & 13.22 (12.43, 14.01) \\ 
\cline{2-7}
10 & 0.1 & 0.50 (0.49, 0.51) & 0.92 (0.88, 0.96) & 1.03 (1.01, 1.05) & 2.38 (2.23, 2.53) & 16.94 (15.61, 19.17) \\ 
\cline{2-7}
 & 0.2 & 0.49 (0.48, 0.50) & 0.94 (0.91, 0.98) & 0.96 (0.92, 0.98) & 2.46 (2.34, 2.63) & 18.84 (17.24, 20.60) \\ 
\cline{1-7}
 & 0.05 & 0.49 (0.48, 0.49) & 0.85 (0.82, 0.88) & 0.98 (0.97, 1.00) & 2.08 (1.99, 2.20) & 13.24 (12.47, 14.26) \\ 
\cline{2-7}
15 & 0.1 & 0.50 (0.49, 0.51) & 0.90 (0.87, 0.95) & 1.02 (1.00, 1.04) & 2.36 (2.25, 2.52) & 16.22 (14.72, 18.07) \\ 
\cline{2-7}
 & 0.2 & 0.50 (0.49, 0.50) & 0.92 (0.89, 0.95) & 0.96 (0.92, 0.99) & 2.42 (2.27, 2.57) & 17.62 (16.08, 19.33) \\ 
\hline \hline
 & 0.05 & 0.50 (0.49, 0.52) & 0.86 (0.84, 0.88) & 1.00 (0.98, 1.02) & 2.13 (2.06, 2.21) & 14.07 (13.58, 14.80) \\ 
\cline{2-7}
5 & 0.1 & 0.51 (0.50, 0.51) & 0.94 (0.91, 0.98) & 1.05 (1.03, 1.08) & 2.46 (2.35, 2.60) & 16.99 (15.77, 18.57) \\ 
\cline{2-7}
 & 0.2 & 0.52 (0.51, 0.54) & 0.95 (0.92, 0.99) & 1.02 (1.00, 1.04) & 2.53 (2.38, 2.73) & 18.72 (17.26, 20.73) \\ 
\cline{1-7}
 & 0.05 & 0.49 (0.49, 0.50) & 0.87 (0.84, 0.90) & 1.01 (0.99, 1.03) & 2.12 (2.04, 2.24) & 13.28 (12.62, 14.02) \\ 
\cline{2-7}
10 & 0.1 & 0.51 (0.50, 0.52) & 0.94 (0.90, 0.98) & 1.05 (1.03, 1.08) & 2.41 (2.30, 2.55) & 17.05 (15.60, 18.85) \\ 
\cline{2-7}
 & 0.2 & 0.50 (0.49, 0.51) & 0.95 (0.93, 0.99) & 0.98 (0.95, 1.01) & 2.50 (2.37, 2.67) & 19.47 (17.74, 21.83) \\ 
\cline{1-7}
 & 0.05 & 0.49 (0.49, 0.50) & 0.87 (0.84, 0.90) & 1.01 (0.99, 1.03) & 2.15 (2.06, 2.25) & 13.62 (12.79, 14.66) \\ 
\cline{2-7}
15 & 0.1 & 0.51 (0.50, 0.52) & 0.92 (0.88, 0.96) & 1.04 (1.02, 1.06) & 2.39 (2.26, 2.58) & 16.69 (15.14, 18.99) \\ 
\cline{2-7}
 & 0.2 & 0.50 (0.50, 0.51) & 0.95 (0.91, 0.98) & 0.98 (0.95, 1.01) & 2.48 (2.35, 2.68) & 18.16 (16.50, 20.29) \\ 
\hline \hline
 & 0.05 & 0.56 (0.52, 0.61) & 1.14 (1.00, 1.33) & 1.24 (1.14, 1.36) & 3.21 (2.82, 3.73) & 22.47 (19.86, 25.53) \\ 
\cline{2-7}
5 & 0.1 & 0.58 (0.54, 0.64) & 1.22 (1.07, 1.40) & 1.31 (1.18, 1.45) & 3.65 (3.10, 4.32) & 26.79 (23.16, 31.02) \\ 
\cline{2-7}
 & 0.2 & 0.58 (0.54, 0.63) & 1.23 (1.09, 1.43) & 1.21 (1.12, 1.30) & 3.64 (3.25, 4.18) & 29.41 (25.54, 34.47) \\ 
\cline{1-7}
 & 0.05 & 0.56 (0.53, 0.60) & 1.11 (0.97, 1.28) & 1.23 (1.13, 1.37) & 3.10 (2.71, 3.54) & 21.82 (19.13, 24.59) \\ 
\cline{2-7}
10 & 0.1 & 0.57 (0.53, 0.63) & 1.19 (1.05, 1.38) & 1.27 (1.15, 1.41) & 3.51 (3.02, 4.17) & 27.49 (23.81, 32.18) \\ 
\cline{2-7}
 & 0.2 & 0.56 (0.52, 0.61) & 1.22 (1.06, 1.37) & 1.14 (1.06, 1.22) & 3.59 (3.14, 4.10) & 30.67 (26.84, 35.73) \\ 
\cline{1-7}
 & 0.05 & 0.56 (0.52, 0.60) & 1.12 (0.99, 1.27) & 1.24 (1.13, 1.37) & 3.17 (2.76, 3.71) & 22.64 (19.72, 25.59) \\ 
\cline{2-7}
15 & 0.1 & 0.58 (0.54, 0.64) & 1.17 (1.02, 1.33) & 1.28 (1.16, 1.42) & 3.49 (3.04, 3.93) & 27.21 (22.89, 31.83) \\ 
\cline{2-7}
 & 0.2 & 0.56 (0.53, 0.61) & 1.19 (1.05, 1.36) & 1.14 (1.04, 1.27) & 3.60 (3.06, 4.35) & 29.11 (24.94, 34.51) \\ 
\hline
\end{tabular}
\end{table}
\clearpage

\section{Applications}
\label{sec:applications}

The following analyses are all performed in \textsf{R}, see \cite{R2019}, with the help of the package \textsf{spatstat}, see \cite{BaddeleyEtAl2015}.

\subsection{Street theft in Bogot\'a}
\label{ssec:bogota}

We investigate a data set of person-related street thefts in Bogot\'a, Colombia, during the years 2012--2017. This data set is part of a huge data set based on a large number of types of crimes collected by the Direcci\'on de Investigaci\'on Criminal e Interpol (DIJIN), a department of the Colombian National Police. We acknowledge DIJIN and the General Santander National Police Academy (ECSAN) for allowing us to use this data. In particular the cases of street theft in Bogot\'a consist of muggings, which involve the use of force or threat, as well as pickpocketing. They do not include theft of vehicles, breaking into cars, etc. Here we focus on the locality of Kennedy, a roughly $7.5\,\mathrm{km} \times 7.5\,\mathrm{km}$ administrative ward in the west of the city, because this area is considered by the police as being more dangerous with a higher average number of crime events compared to the rest of Bogot\'a. The total number of street thefts in Kennedy for the considered period is 25\hbit 840.

Since a plot of weekly numbers of crimes reveals no clear seasonal pattern and since weekly patterns (and hence their barycenters) are of a good size to be interpreted graphically, we compute yearly barycenters for these weekly patterns. Thus we may think of a barycenter pattern as representing a ``typical week'' of street thefts in the corresponding year. As penalty parameter we chose 1000\,m. Since street information was not directly available to us, we chose Euclidean distance as a metric and set $p=2$ to be able to relate to our simulation results in the previous section.

Each barycenter was computed based on 100 starting patterns with cardinalities regularly scattered over the integer numbers between the 0.45 to the 0.7 quantiles of the weekly number of data points for the corresponding year. We chose this somewhat asymmetrically around the median, because the mean number of thefts (the theoretical number of points in the barycenter if the penalty becomes large) was typically quite a bit to the right of the median, and also because our algorithm is somewhat better at deleting than at adding points.

Figure~\ref{fig:bogota_barys} depicts the obtained barycenters, which except for the last pattern have cardinalities just slightly below the average weekly numbers of muggings of 51.7, 57.6, 52.8, 54.4, 82.5 and 196.9, respectively. The barycenters for the years 2012--2015 seem to be largely similar. Then in 2016 we start seeing patterns of denser structures forming along a line to the west and a center in the south-east of Kennedy. These can be actually identified as a main street and a major intersection in the densely populated parts of Kennedy.

\begin{figure}[h]
  \includegraphics[width=1.02\linewidth]{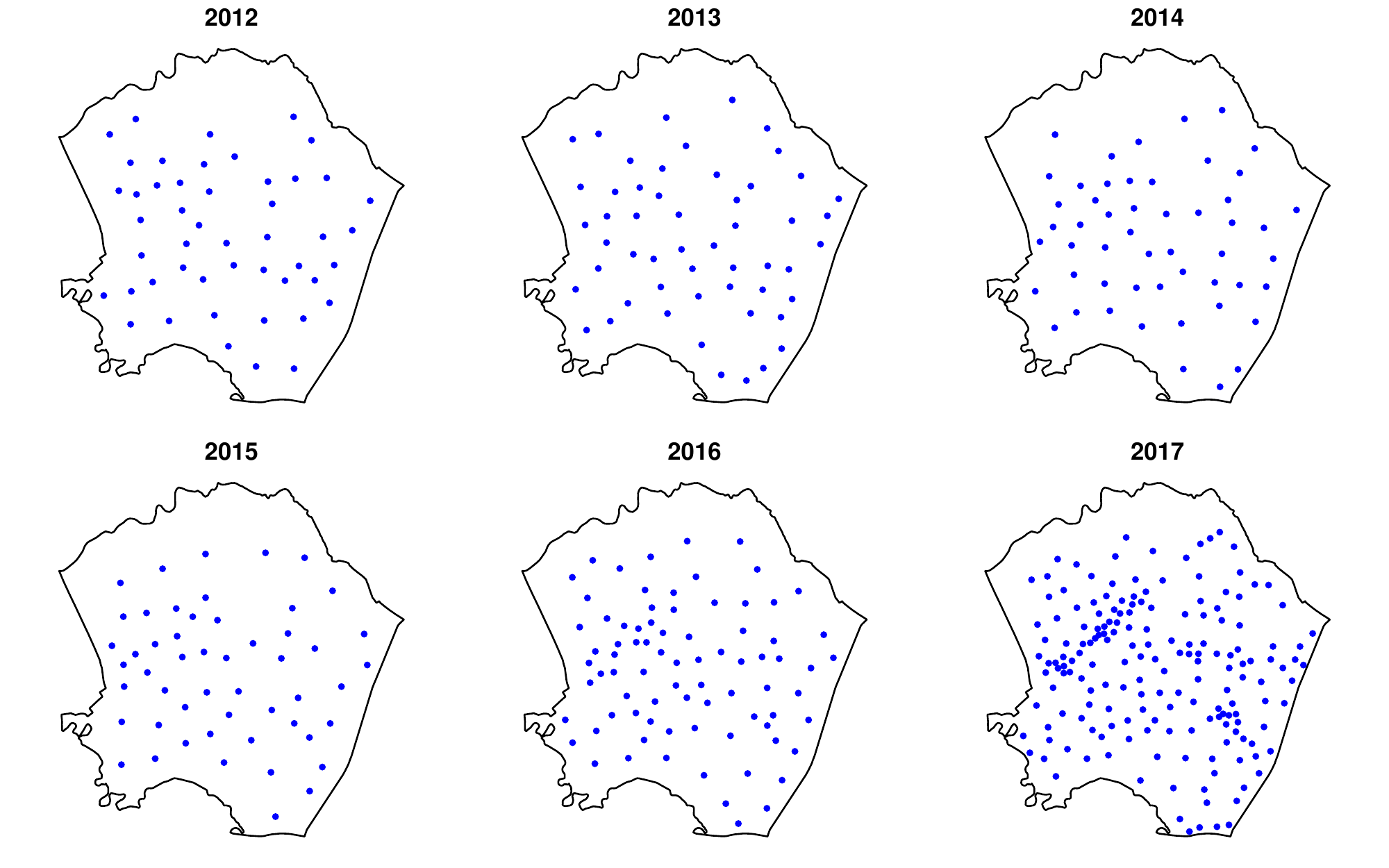}
\vspace*{-8mm}
  
\caption{Barycenters of weekly street thefts in the localidad of Kennedy in Bogot\'a. The cardinalities are 48, 53, 52, 52, 80 and 175, respectively.}
\label{fig:bogota_barys} 
\end{figure}

\subsection{Assault cases in Valencia}
\label{ssec:valencia}

As a second application we analyze cases of assault in Valencia, Spain, reported to the police in the years 2010--2017. Since the addresses of the assaults and the street network are available, we treat this data as point patterns on a graph using shortest-path distance and $p=1$. We acknowledge the local police in Valencia city together with the 112 emergency phone that kindly provided us the data after cleaning and removing any personal information.

We split up the graph and analyze the four central districts of Ciutat Vella, Eixample, Extramurs and El Pla del Real separately. For this we assigned each assault case to its district, but added also streets from other districts at the boundary, in order to enable more natural shortest-path computations. The north-south and east-west extensions of the districts vary roughly between 1.6 and 3.3\,km.

In the time domain we split up the assault data by year and season into seven winter patterns (data from December, January and February) and eight summer patterns (data from June, July and August), discarding for the present analysis data from the intermediate seasons, as well as from January and February 2010 and December 2017. We then computed barycenters per main season and district, obtaining ``typical'' assault patterns for summer and winter for each of the four districts considered, see Figure~\ref{fig:projbarycentersraw}. The penalty was chosen as 800\,m with respect to the shortest-path distance.

As mentioned in Subsection~\ref{ssec:original} when describing the subroutine \optimBary that finds cluster centers on networks, we can calculate all distances that are relevant to the algorithm beforehand. For this we use the corresponding functionality built into the \texttt{linnet} objects in \textsf{spatstat}, which is very fast for the present purpose. On a standard laptop with a 1.6 GHz Intel i5 processor the computation took only about four seconds for the largest data set, which is Ciutat Vella in summer with a total of $2494$ vertices (1676 street crossings plus 818 data points).

For the starting patterns in each district in summer and in winter, we chose $n$ random points, where $n$ ranged from $0.8$ times to $1.15$ times the median cardinality of the data point patterns. Since our present implementation of the {\kmeansbary} algorithm on graphs runs without \optimDelete and \optimAdd steps, we based each barycenter on a large sample of $500$ starting patterns for each $n$. This resulted in an overall total of $101\hbit500$ calls to our algorithm for the eight scenarios, which on average took $0.57$ seconds each, using the precomputed distance matrices. One calculation in the largest setting (Ciutat Vella in summer) takes about $0.82$ seconds and in the smallest (El Pla del Real in summer) about $0.08$ seconds. The increase in the objective function was only up to 1\% when decreasing the total number of calls to our algorithm by a factor of $20$, resulting in a total computation time of well under one hour.

Due to the choice of $p=1$ it happens quite frequently that there are several optimal centers for some of the clusters obtained after convergence of the \kmeansbary algorithm. In this case we take the average of their coordinates and project the result back onto the graph in order to obtain a somewhat more balanced result. The resulting point does not necessarily realize the same cluster cost as the original center points, but on a real street network it is not to be expected that the cost becomes considerably worse. In fact, for the data considered, the results hardly differed at all.

Figure~\ref{fig:projbarycentersraw} shows the barycenters for the different districts in winter and in summer. It seems that there are no very clear effects of the season on the assaults. In the first district (Ciutat Vella) there are substantially more assaults in summer, but their spatial distribution in the barycenter is more or less similar. In Eixample we see a concentration of assault cases in summer in the Barrio Ruzafa in the southern half of the district, whereas cases are more or less equally spread in winter. A notable feature is the occurrence of 23 barycenter points at a single crossing in summer and 7 points at the same crossing in winter with further points close by. This is due to the cul-de-sac visible in Figure~\ref{fig:projbarycentersraw}, which in reality forms sort of a backyard that makes the area easy for assaults, especially in summer time when there are more people (especially tourists) moving around those parts of the city. The spot is well-known to the police and in recent years the number of assaults has decreased due to police interventions. The barycenters clearly reflect this (former) assault hot spot.

In the district of Extramurs both barycenters are more or less spread over the whole district, with two clusters of assaults occurring in the east and south. Both clusters are much more pronounced in winter. In the district of El Pla del Real there is some concentration in the winter month in the east and south-east. Apart from that the only noticeable difference is that there are substantially more assaults in winter than in summer, which may well be related to the fact that this is a popular student district.
\begin{figure}[h]
  \includegraphics[scale=0.75]{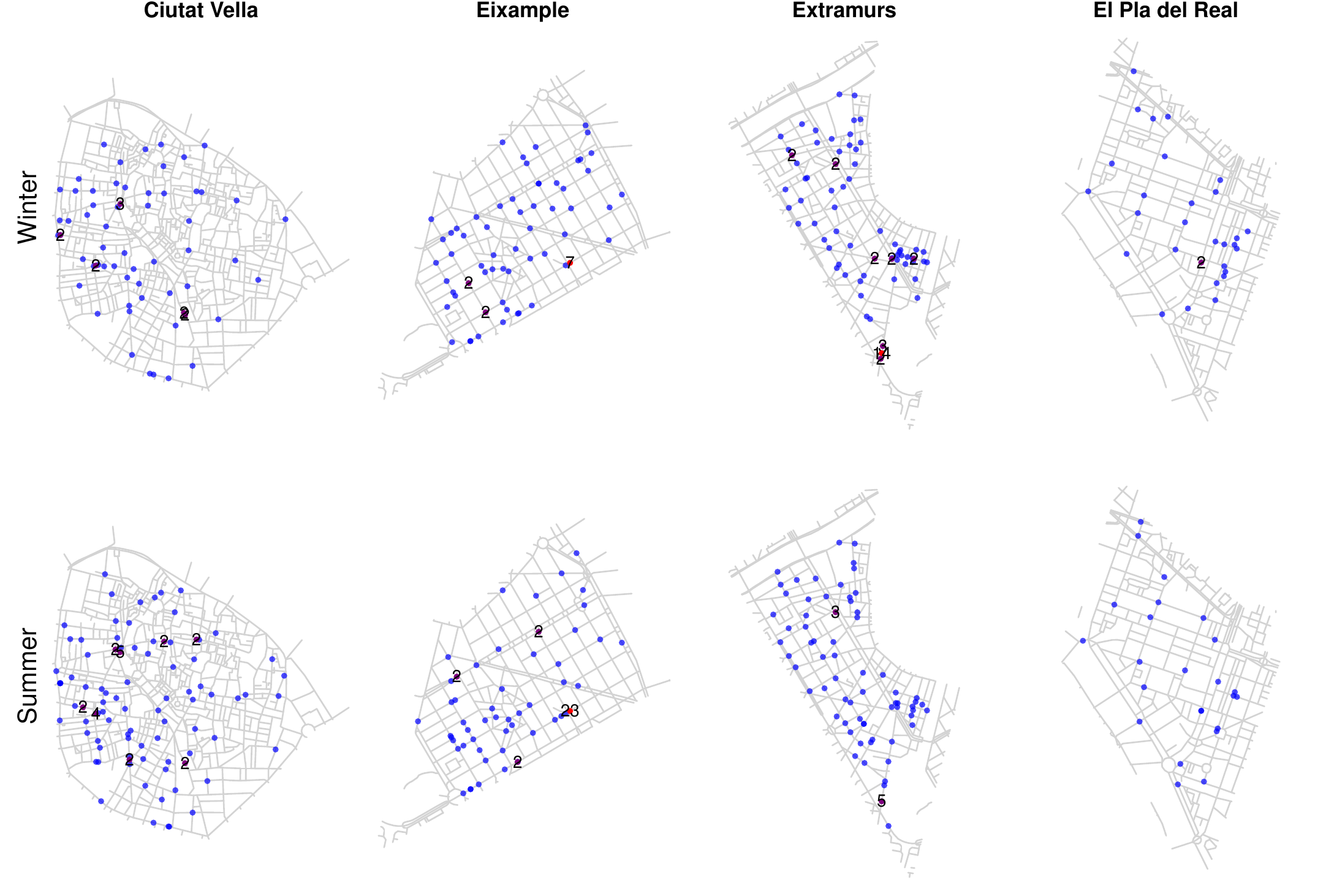}
\vspace*{-7mm}
  
\caption{Barycenters of cases of assault for different districts of Valencia in winter and in summer. The numbers indicate multiplicities if there are several points at a single location. The cardinalities of the barycenters are 68, 69, 88, 30 (winter) and 103, 79, 74, 24 (summer).}
\label{fig:projbarycentersraw} 
\end{figure}

\section{Discussion and outlook}

In this paper we have introduced the $p$-th order TT- and RTT-metrics, which allow us to measure distances between point patterns in an intuitive way, generalizing several earlier metrics. We have investigated $q$-th order barycenters with respect to the TT-metric and presented two variants of a heuristic algorithm. These variants return local minimizers of the Fr\'echet functional that mimic properties of the actual barycenter well and attain consistent objective function values. They are computable in a few seconds for medium-sized problems, such as 100 patterns of 100 points.

For the proof of Proposition~\ref{prop:mdagengen} it was necessary to set $p=q$. While such a choice may seem natural, we point out that due to the separate interpretations of $p$ as the order for matching points in the metric on $\mfnfin$ (higher $p$ tends to balance out the matching distances) and $q$ as the order of the empirical moment in $\mfnfin$, it may well be desirable to combine $p \neq q$. 

In the present paper we have only dealt with the descriptive aspects of barycenters. It is thus clear that our applications in Section~\ref{sec:applications} can only be seen as explorative studies. In order to determine whether differences between group barycenters are statistically significant, we need to take the distribution of the point patterns around their barycenters into account and perform appropriate hypothesis tests.

Fortunately, the Fr\'echet functional~\eqref{eq:barydef} provides us with a natural quantification of scatter around the barycenter. For $q=2$ it is quite common to refer to
\begin{equation*}
  \var(\xi_1,\ldots,\xi_k) = \min_{\zeta \in \mfnfin} \frac{1}{k} \sum_{j=1}^{k} \tau(\xi_j,\zeta)^2
\end{equation*}
as (empirical) Fr\'echet variance, due to Equation~\eqref{eq:frechet2}. Detailed asymptotic theory for performing an analysis of variance (ANOVA) in metric spaces based on comparing Fr\'echet variances has been recently developed in \cite{DubeyMueller2017}. The application and adaptation of this theory for the point pattern space and an investigation of the performance of our heuristic algorithm in this context will be the subject of a future paper.

In a similar vein, a fast computation of reliable barycenters opens many doors to more advanced procedures in statistics and machine learning. This includes barycenter-based dimension reduction techniques, such as Wasserstein dictionary learning, see \cite{SchmitzEtAl2018}, and functional principal component analysis of point patterns evolving in time, see \cite{DubeyMueller2019}.

\section*{Appendix: Proofs left out in the main text}

\renewcommand{\theequation}{A.\arabic{equation}}
\renewcommand{\thethm}{A.\arabic{thm}}

\begin{lem}
\label{lem:metric}
  Let $C > 0$, $\widetilde{C} \in (0,2C]$ and let $(\mcx,d)$ be a metric space with $\diam(\mcx) = \sup_{x,y \in \mcx} d(x,y) \leq 2C$. For $k \in \NN$ set $\mcx' = \mcx \cup \{\aleph_1,\ldots,\aleph_k\}$, where $\aleph_1,\ldots \aleph_k \not\in \mcx$ are pairwise different, and define
  \begin{equation*}
    d'(x,y) = \begin{cases}
      d(x,y) &\text{if $x,y \in \mcx$;} \\
      C &\text{if $\{x,y\} \cap \mcx \neq \emptyset$ and $\{x,y\} \cap \{\aleph_1,\ldots,\aleph_k\} \neq \emptyset$;} \\
      \widetilde{C} &\text{if $\{x,y\} \subset \{\aleph_1,\ldots,\aleph_k\}$ and $x \neq y$;} \\
      0 &\text{if $x=y=\aleph_i$ for some $i \in [k]$.}
    \end{cases}
  \end{equation*}
  Then $(\mcx',d')$ is a metric space.
\end{lem}
\begin{proof}
  Identity and symmetry properties of the map $d' \colon \mcx' \times \mcx' \to \Rplus$ follow immediately.
  Since $d$ is a metric on $\mcx$ and $\tilde{d}(x,y) = \widetilde{C} \one\{x \neq y\}$ is a metric on $\mcy = \{\aleph_1,\ldots,\aleph_k\}$, 
  we only have to check the triangle inequality for a few special cases. If $x \in \mcx$ and $y \in \mcy$ (or vice versa), then $d(x,y) = C$. Since one of $d(x,z)$ and $d(z,y)$ has to be $=C$ regardless of $z \in \mcx'$, we obtain $d(x,y) \leq d(x,z) + d(z,y)$. If $x,y \in \mcx$ and $z \in \mcy$, then
  \begin{equation*}
    d(x,y) \leq \diam(\mcx) \leq 2C = d(x,z) + d(z,y).
  \end{equation*}
  Likewise, if $x,y \in \mcy$ and $z \in \mcx$, then
  \begin{equation*}
    d(x,y) \leq \widetilde{C} \leq 2C = d(x,z) + d(z,y).
  \end{equation*}  
  \vspace*{-8mm}
  
\end{proof}

\begin{proof}[Proof of Theorem~\ref{thm:ttassignment}]
  Since $\bartau(\xi,\eta) = \frac{1}{n^{1/p}} \tau(\xi,\eta)$, it is enough to show the statement for $\tau$.
  
  Let $\pi \in S_n$ be a permutation that minimizes $\sum_{i=1}^n d'(x_i,y_{\pi(i)})^p$. Writing $I = \bigl\{i \in [m];\, d(x_i,y_{\pi(i)}) < 2^{1/p} C \bigr\}$ we obtain
  \begin{equation}
    \label{eq:dprimetable}
    d'(x_i, y_{\pi(i)}) =
    \begin{cases}
      d(x_i,y_{\pi(i)}) &\text{if $i \in I$}; \\
      2^{1/p}C            &\text{if $i \in [m] \setminus I$}; \\
      C                   &\text{if $i \in [n] \setminus [m]$}.
    \end{cases}
  \end{equation}
  Therefore, enumerating $I$ in arbitrary order as $\{i_1, \ldots,i_l\}$ for some $l \in [m]$ and setting $j_r := \pi(i_r)$ for $r \in [l]$, we have
  \begin{equation}
    \label{eq:ttassign_equal}
  \begin{split}
    \sum_{i=1}^n d'(x_i,y_{\pi(i)})^p &= \sum_{r=1}^{l} d(x_{i_r},y_{j_r})^p + (m-l) (2^{1/p}C)^p + (n-m) C^p \\
    &= (m+n-2l) C^p + \sum_{r=1}^{l} d(x_{i_r},y_{j_r})^p.
  \end{split}
\end{equation}
Thus $\tau(\xi,\eta)^p \leq \min_{\pi \in S_n} \sum_{i=1}^n d'(x_i,y_{\pi(i)})^p$.

Conversely, let $(i_1,\ldots,i_l; j_1,\ldots,j_l) \in S(m,n)$ minimize $(m+n-2l) C^p + \sum_{r=1}^{l} d(x_{i_r},y_{j_r})^p$. This implies $d(x_{i_r},y_{j_r}) \leq 2^{1/p} C$ for all $r \in [l]$, because otherwise we could obtain a smaller value by removing $i_r$, $j_r$ from the vector. Writing $I = \{i_1, \ldots,i_l\}$, $J = \{j_1, \ldots,j_l\}$ it implies also that $d(x_{i},y_{j}) \geq 2^{1/p} C$ for all $i \in [m] \setminus I$ and $j \in [n] \setminus J$, because otherwise we could obtain a smaller value by adding $i$, $j$ to the vector. Let then $\pi \in S_n$ be any permutation satisfying $\pi(i_r) = \pi(j_r)$ for all $r \in [l]$. With this $\pi$ we obtain exactly the $d'$-distances in~\eqref{eq:dprimetable} for all $i \in [n]$ and hence \eqref{eq:ttassign_equal} holds again. Thus $\min_{\pi \in S_n} \sum_{i=1}^n d'(x_i,y_{\pi(i)})^p \leq \tau(\xi,\eta)^p$.
\end{proof}

\begin{proof}[Proof of Proposition~\ref{prop:metrics}]
  We start with the map $\tau \colon \mfnfin \times \mfnfin \to \Rplus$. If $\xi=\eta$, then $m=n$ and there is a permutation $\pi \in S_n$ such that $x_i = y_{\pi(i)}$ for $1 \leq i \leq n$. Hence $\tau(\xi,\eta)=0$, choosing $l=n$ and $(i_1,\ldots, i_l; j_1,\ldots,j_l) = (1,\ldots,n,\pi(1),\ldots,\pi(n))$. If on the other hand $\tau(\xi,\eta) = 0$, we must have $l=m=n$ to be able to achieve $m+n-2l=0$ and there must be $(i_1,\ldots,i_n; j_1,\ldots,j_n) \in S(n,n)$ such that $d(x_{i_r},y_{j_r})=0$ for $1 \leq r \leq n$. Since the $d$ is a metric, this yields $\xi = \sum_{r=1}^n \delta_{i_r} = \sum_{r=1}^n \delta_{j_r} = \eta$. The symmetry of $\tau$ is immediately clear from the symmetric form of \eqref{eq:ttdef}.

  For the proof of the triangle inequality we use the metric space $(\mcx',d')$ introduced before Theorem~\ref{thm:ttassignment}. Let $\xi,\eta,\zeta \in \mfnfin$. After filling up patterns to the maximum of the three cardinalities by adding points at the auxiliary location $\aleph$, we may assume that $\xi = \sum_{i=1}^n \delta_{x_i}$, $\eta = \sum_{j=1}^n \delta_{y_j}$ and $\zeta = \sum_{k=1}^n \delta_{z_k}$ have the same cardinality. Noting that given two point patterns of the same cardinality we may add any number of extra points located at $\aleph$ to both of them without changing their $\tau$-distance, Theorem~\ref{thm:ttassignment} yields that there are $\pi_1,\pi_2 \in S_n$ such that
  \begin{equation*}
    \tau(\xi,\zeta) = \biggl( \sum_{i=1}^n d'(x_i,z_{\pi_1(i)})^p \biggr)^{1/p} \; \text{ and} \quad
    \tau(\zeta,\eta) = \biggl( \sum_{i=1}^n d'(z_i,y_{\pi_2(i)})^p \biggr)^{1/p}. \quad
  \end{equation*}
  Then $\pi = \pi_2 \circ \pi_1 \in S_n$ matches the points of $\xi$ and $\eta$ in such a way that
  \begin{equation*}
    d'(x_i,y_{\pi(i)}) \leq d'(x_i,z_{\pi_1(i)}) + d'(z_{\pi_1(i)},y_{\pi_2(\pi_1(i))})
  \end{equation*}
  and Theorem~\ref{thm:ttassignment} and the triangle inequality for the $\ell_p$-norm yields that
  \begin{equation*}
  \begin{split}
    \tau(\xi,\eta) &\leq \biggl( \sum_{i=1}^n d'(x_i,y_{\pi(i)})^p \biggr)^{1/p} \\
      &\leq \biggl( \sum_{i=1}^n d'(x_i,z_{\pi_1(i)})^p \biggr)^{1/p}
      + \biggl( \sum_{i=1}^n d'(z_{\pi_1(i)},y_{\pi_2(\pi_1(i))})^p \biggr)^{1/p} \\
      &= \tau(\xi,\zeta) + \tau(\zeta,\eta).
  \end{split}
\end{equation*}

We turn to the map $\bartau \colon \mfnfin \times \mfnfin \to \Rplus$.
Since $\bartau(\xi,\eta) = \frac{1}{\max\{\abs{\xi},\abs{\eta}\}^{1/p}} \tau(\xi,\eta)$, we may inherit the identity and symmetry properties for $\bartau$ directly from $\tau$. To show the triangle inequality, let $\xi = \sum_{i=1}^{m_1} \delta_{x_i}$, $\eta = \sum_{j=1}^{m_2} \delta_{y_j}$ and $\zeta = \sum_{k=1}^n \delta_{z_k}$ be in $\mfnfin$ and set $m_* = \max\{m_1,m_2\}$. If $n \leq m_*$, we obtain the desired result from the triangle inequality of~$\tau$ as
\begin{equation*}
\begin{split}
  \bartau(\xi,\eta) &= \frac{1}{m_*^{1/p}} \tau(\xi,\eta) \\
  &\leq \frac{1}{m_*^{1/p}} \bigl( \tau(\xi,\zeta) + \tau(\zeta,\eta) \big) \\
  &\leq \frac{1}{\max\{m_1,n\}^{1/p}} \tau(\xi,\zeta) + \frac{1}{\max\{n,m_2\}^{1/p}} \tau(\zeta,\eta) \\[1mm]
  &\leq \bartau(\xi,\zeta) + \bartau(\zeta,\eta).
\end{split}
\end{equation*}

If $n > m_*$, we use a slightly different construction for the extended metric space. Let $\mcx' = \mcx \cap \{\aleph,\aleph'\}$ for two different $\aleph, \aleph' \not\in \mcx$. Setting
\begin{equation*}
  d'(x,y) = \begin{cases}
    \min\{d(x,y), 2^{1/p} C\} &\text{if $x,y \in \mcx$},\\
    C                         &\text{if $\{x,y\} \cap \mcx \neq \emptyset$ and $\{x,y\} \cap \{\aleph,\aleph'\} \neq \emptyset$},\\
    2^{1/p}C                  &\text{if $\{x,y\} = \{\aleph,\aleph'\}$},\\
    0                         &\text{if $x=y=\aleph$ or $x=y=\aleph'$},
  \end{cases}
\end{equation*}
we obtain by Lemma~\ref{lem:metric} that $(\mcx',d')$ is again a metric space.
Setting $x_i = \aleph$ for $m_1+1 \leq i \leq n$ and $y_j = \aleph'$ for $m_2+1 \leq j \leq n$, we may define $\txi = \sum_{i=1}^n \delta_{x_i}$ and $\teta = \sum_{j=1}^n \delta_{y_j}$. Note that an optimal permutation $\pi_* \in S_{m_*}$ for $\bartau(\xi,\eta)$ can be extended to an optimal permutation $\tpi_* \in S_n$ for $\bartau(\txi,\teta)$ by setting $\tpi_*(i) = i$ for $m_*+1 \leq i \leq n$. Furthermore, for any $s,c \geq 0$ with $s \leq m_* c$, we have $\frac{1}{m_*}s \leq \frac{1}{n} \bigl( s+(n-m_*)c \bigr)$. Combining these two facts,
we obtain
\begin{equation*}
  \bartau(\xi,\eta)^p =  \frac{1}{m_*} \sum_{i=1}^{m_*} d'(x_i,y_{\pi_*(i)})^p 
  \leq \frac{1}{n} \biggl( \sum_{i=1}^{m_*} d'(x_i,y_{\pi_*(i)})^p + (n-m_*) \cdot 2C^p \biggr) = \bartau(\txi,\teta)^p,
\end{equation*}
and therefore
\begin{equation*} 
  \bartau(\xi,\eta) \leq \bartau(\txi,\teta) \leq \bartau(\txi,\zeta) + \bartau(\zeta, \teta) = \bartau(\xi,\zeta) + \bartau(\zeta, \eta),
\end{equation*}
where the second inequality holds since the cardinalities of all point patterns are equal and the equality holds by two more applications of Theorem~\ref{thm:ttassignment}.
\end{proof}

\begin{proof}[Proof of Proposition~\ref{prop:comparison}]
  The equivalence for (a) was already used in~\cite{DiezEtAl2012}. We give a quick argument for the sake of completeness.
  We may assume without loss of generality that, in an admissible path $P=(\xi_0,\ldots,\xi_N)$ for the minimization problem~\eqref{eq:spiketime},
  \begin{tightitemize}
  \item only moves from $x \in \xi$ to $y \in \eta$ occur;
  \item only points $y \in \eta$ are added;
  \item only points $x \in \xi$ are deleted;
  \end{tightitemize}
  because if any of these conditions were violated, the total cost of the path could only become larger (for the first item we use the triangle inequality for $d$). The minimization~\eqref{eq:spiketime} is then equivalent to choosing $l \in \{0,1,\ldots,\min\{m,n\}\}$ points to be moved from $\xi$-points with indices $i_1,\ldots,i_l \in [m]$ to $\eta$-points with indices $j_1,\ldots,j_l \in [n]$, respectively, at cost $d(x_{i_r},y_{j_r})$ for each move. The remaining $m-l$ points of $\xi$ are deleted at cost $C$ per deletion, and the remaining $n-l$ points of $\eta$ are added at cost $C$ per addition. This yields exactly the minimization problem~\eqref{eq:ttdef}.

  The equivalence (b) is an immediate consequence of Theorem~\ref{thm:ttassignment}.
\end{proof}

\end{document}